\pdfoutput=1
\documentclass[paper=a4,DIV=12,11pt]{scrartcl}
\usepackage[USenglish]{babel}
\addtokomafont{disposition}{\rmfamily}
\addtokomafont{descriptionlabel}{\rmfamily\boldmath}
\usepackage[utf8]{inputenc}
\usepackage{amsmath}
\usepackage{enumitem}
\usepackage{listings}
\usepackage{booktabs}
\usepackage{calc}

\usepackage{xcolor}
\usepackage{tikz}
\usepackage{flafter}

\lstdefinelanguage{pseudo}{%
  morekeywords={if,then,else,endif,while,do,done,repeat,until,for,foreach,to,downto,do,end,break,and,or,fail,return,output,result,input,true,false},
  sensitive=false,
  morecomment=[l]{//},
  morecomment=[s]{/*}{*/},
  mathescape=true,
}
\newlength{\listingbelowcaptionskip}
\usepackage{microtype}

\usepackage{etoolbox}

\newcommand{\probitem}{\normalfont\normalcolor\mathversion{normal}}
\newenvironment{problemdef}[1]{%
  \begin{description}
   \item[{\probitem #1:}]
}{%
  \end{description}
}

\newcommand{\assign}{\leftarrow}
\newcommand{\algo}[1]{\texttt{\upshape #1}}

\usepackage{amssymb}
\usepackage{amsfonts}
\usepackage{amsthm}
\usepackage{xspace}
\usepackage{csquotes}

\makeatletter
\def\Ifstartswithcite#1#2\Endlist#3#4{%
  \ifstrequal{\cite}{#1}{#3}{#4}%
}
\patchcmd{\thmhead@plain}{(#3)}{%
  \expandafter\Ifstartswithcite\@firstofone#3\@empty\@empty\Endlist{#3}{(#3)}%
}{}{}%
\let\thmhead\thmhead@plain
\makeatother

\let\emptyset\varnothing
\let\phi\varphi

\newcommand{\cclass}[1]{\textsf{\textup{#1}}}
\newcommand{\FPT}{\cclass{FPT}\xspace}
\newcommand{\Wone}{\cclass{W[1]}\xspace}
\newcommand{\NP}{\cclass{NP}\xspace}

\theoremstyle{plain}
\newtheorem*{claim}{Claim}
\newtheorem{theorem}{Theorem}[section]
\newtheorem{lemma}[theorem]{Lemma}
\newtheorem{corollary}[theorem]{Corollary}
\theoremstyle{definition}
\newtheorem{definition}[theorem]{Definition}
\theoremstyle{remark}
\newtheorem{remark}[theorem]{Remark}
\theoremstyle{plain}

\newenvironment{proof*}[1]{\def\proofname{#1}\begin{proof}}{\end{proof}}



\newenvironment{acknowledgments}
  {\par\vskip 6pt\noindent\textbf{Acknowledgments.}}{\par}
\usepackage{color}
\usepackage[pdfstartview=FitH,pdfpagemode=UseNone,colorlinks=true,citecolor=blue,linkcolor=blue]{hyperref}

\usepackage{mathtools}
\DeclarePairedDelimiter{\card}{\lvert}{\rvert}
\DeclarePairedDelimiter{\paren}{\lparen}{\rparen}
\let\angle\undefined
\DeclarePairedDelimiter{\angle}{\langle}{\rangle}

\makeatletter
\DeclarePairedDelimiter{\SK@setone}{\lbrace}{\rbrace}
\DeclarePairedDelimiterX{\SK@settwo}[2]{\lbrace}{\rbrace}{#1\:\delimsize\vert\:#2}
\newcommand{\set}{\@ifstar{\SK@set{*}}{\SK@oset}}
\newcommand{\SK@oset}[1][]{\ifstrempty{#1}{\SK@set{}}{\SK@set{[#1]}}}
\newcommand{\SK@set}[2]{\@gifnextchar\bgroup{\SK@@set{#1}{#2}}{\SK@setone#1{#2}}}
\newcommand{\SK@@set}[3]{\ifstrempty{#1}{%
    \SK@settwo{#2}{#3}%
  }{%
    \SK@settwo#1{#2}{\begin{array}{@{}l@{}}#3\end{array}}%
  }}
\def\@gifnextchar#1#2#3{\let\@tempe#1\def\@tempa{#2}\def\@tempb{#3}%
  \futurelet\@tempc\@gifnch}
\def\@gifnch{\ifx\@tempc\@sptoken\let\@tempd\@tempb%
  \else\ifx\@tempc\@tempe\let\@tempd\@tempa\else\let\@tempd\@tempb\fi\fi\@tempd}
\makeatother

\DeclareMathOperator{\support}{supp}
\DeclareMathOperator{\compl}{compl}
\DeclareMathOperator{\Sym}{Sym}
\DeclareMathOperator{\Alt}{Alt}
\DeclareMathOperator{\id}{id}
\DeclareMathOperator{\Ker}{Ker}
\DeclareMathOperator{\poly}{poly}
\DeclareMathOperator{\Iso}{Iso}
\DeclareMathOperator{\Aut}{Aut}

\DeclareMathOperator{\CG}{CG}
\DeclareMathOperator{\blocks}{\mathcal{B}}
\DeclareMathOperator{\Var}{Var}

\newcommand{\prob}[1]{\textsc{\mdseries\rmfamily #1}}

\newcommand{\GI}{\prob{GI}\xspace}

\newcommand{\ExactCNFGA}{\prob{Exact-CNF-GA}\xspace}
\newcommand{\ExactCNFHGA}{\prob{Exact-CNF-HGA}\xspace}
\newcommand{\ExactCNFGI}{\prob{Exact-CNF-GI}\xspace}
\newcommand{\ExactCNFHGI}{\prob{Exact-CNF-HGI}\xspace}
\newcommand{\CNFHGI}{\prob{CNF-HGI}\xspace}

\newcommand{\BLUE}{\textsc{Blue}\xspace}
\newcommand{\RED}{\textsc{Red}\xspace}
\newcommand{\NEW}{\textsc{New}\xspace}

\newcommand{\colGA}{\prob{Col-GA}\xspace}

\newcommand{\Nset}{\ensuremath{\mathbb{N}}}

\usepackage{authblk}

\title{Finding Small Weight Isomorphisms with Additional Constraints is
  Fixed-Parameter Tractable\footnote{An extended abstract of this article
    appears in the proceedings of IPEC~2017.}}

\author[1]{V.~Arvind}
\author[2]{Johannes~K\"obler}
\author[2]{Sebastian~Kuhnert}
\author[3]{Jacobo~Tor\'an}
\affil[1]{Institute of Mathematical Sciences (HBNI), Chennai, India\\
    \texttt{arvind@imsc.res.in}}  
\affil[2]{Institut f\"ur Informatik, Humboldt-Universit\"at zu Berlin, Germany\\
    \{\texttt{koebler},~\texttt{kuhnert}\}\texttt{@informatik.hu-berlin.de}}
\affil[3]{Institut f\"ur Theoretische Informatik, Universit\"at Ulm, Germany\\
    \texttt{toran@uni-ulm.de}}
\date{\vspace{-2.5em}}

\begin{document}

\maketitle

\begin{abstract}
  \noindent\textbf{Abstract.}
  Lubiw showed that several variants of Graph Isomorphism are \NP-complete,
  where the solutions are required to satisfy certain additional
  constraints~\cite{Lub81}. One of these, called \prob{Isomorphism With
    Restrictions}, is to decide for two given graphs $X_1=(V,E_1)$ and
  $X_2=(V,E_2)$ and a subset $R\subseteq V\times V$ of forbidden pairs whether
  there is an isomorphism~$\pi$ from~$X_1$ to~$X_2$ such that $i^\pi\ne j$ for
  all $(i,j)\in R$. We prove that this problem and several of its
  generalizations are in fact in~$\FPT$:
  \begin{itemize}
   \item The problem of deciding whether there is an isomorphism between two
    graphs that moves $k$~vertices and satisfies Lubiw-style constraints is
    in~\FPT, with $k$~and~$\card{R}$ as parameters. The problem remains in~\FPT
    even if a CNF of such constraints is allowed. As a
    consequence of the main result it follows that the problem to decide whether
    there is an isomorphism that moves exactly $k$~vertices is in~\FPT. This
    solves a question left open in~\cite{AKKT17}.

   \item When the weight and complexity are unrestricted, finding isomorphisms
    that satisfy a CNF of Lubiw-style constraints is in $\FPT^{\prob{GI}}$.
    
   \item Checking if there is an isomorphism between two graphs that has
    complexity~$t$ is also in~\FPT with~$t$ as parameter, where the complexity
    of a permutation~$\pi$ is the Cayley measure defined as the minimum
    number~$t$ such that $\pi$~can be expressed as a product of
    $t$~transpositions.

   \item We consider a more general problem in which the vertex set of a
    graph~$X$ is partitioned into \RED and \BLUE, and we are
    interested in an automorphism that stabilizes \RED and \BLUE
    and moves exactly~$k$ vertices in~\BLUE, where $k$ is the parameter.
    This problem was introduced in~\cite{DoFe}, and in~\cite{AKKT17} we showed
    that it is \Wone-hard even with color classes of size~$4$
    inside~\RED. Now, for color classes of size at most~$3$ 
    inside~\RED, we show the problem is in~\FPT.
  \end{itemize}
  In the non-parameterized setting, all these problems are \NP-complete. Also,
  they all generalize in several ways the problem to decide whether there is an
  isomorphism between two graphs that moves at most~$k$ vertices, shown to be in
  \FPT by Schweitzer~\cite{Schw}.
\end{abstract}

\section{Introduction}\label{sec:intro}

The Graph Isomorphism problem (\GI) consists in deciding whether two given input
graphs are isomorphic, i.e., whether there is a bijection between the vertex
sets of the two graphs that preserves the adjacency relation. It is an
intensively researched algorithmic problem for over four decades, culminating in
Babai's recent quasi-polynomial time algorithm~\cite{Ba16}.

There is also considerable work on the parameterized complexity of~\GI. For
example, already in~1980 it was shown~\cite{FHL80} that \GI, parameterized by
color class size, is fixed-parameter tractable (FPT). It is also known that~\GI,
parameterized by the eigenvalue multiplicity of the input graph, is
in~\FPT~\cite{BGM82}. More recently, \GI, parameterized by the
treewidth of the input graph, is shown to be in~\FPT~\cite{LPPS17}.

In a different line of research, Lubiw~\cite{Lub81} has considered the
complexity of~\GI with additional constraints on the isomorphism. Exploring the
connections between~\GI and the \NP-complete problems, Lubiw defined the
following version of~\GI.
\begin{problemdef}{\prob{Isomorphism With Restrictions}}
  Given two graphs $X_1=(V_1,E_1)$ and $X_2=(V_2,E_2)$ and a set of forbidden
  pairs $R\subseteq V_1\times V_2$, decide whether there is an isomorphism~$\pi$
  from~$X_1$ to~$X_2$ such that $i^\pi\ne j$ for all $(i,j)\in R$.
\end{problemdef}
When $X_1=X_2$, the problem is to check if there is an automorphism that
satisfies these restrictions. Lubiw showed that the special case of testing for
\emph{fixed-point-free automorphisms} is \NP-complete. Klav{\'i}k et~al.\
recently reexamined \prob{Isomorphism With Restrictions}~\cite{KKZ16}. They show
that it remains \NP-complete when restricted to graph classes for which \GI is
as hard as for general graphs. Conversely, they show that it can be solved in
polynomial time for several graph classes for which the isomorphism problem is
known to be solvable in polynomial time by combinatorial algorithms, e.g.~planar
graphs and bounded treewidth graphs. However, they also show that the problem
remains \NP-complete for bounded color class graphs, where an efficient group
theoretic isomorphism algorithm is known.

A different kind of constrained isomorphism problem was introduced
by Schweitzer~\cite{Schw}. The weight (or support size) of a permutation
$\pi\in\Sym(V)$ is $\card{\set{i\in V}{i^\pi\ne i}}$. Schweitzer showed that the problem of
testing if there is an isomorphism~$\pi$ of weight at most~$k$ between two
$n$-vertex input graphs in the same vertex set can be solved in time $k^{\mathcal{O}(k)}\poly(n)$. Hence, the
problem is in \FPT with~$k$ as parameter. Schweitzer's algorithm exploits
interesting properties of the structure of an isomorphism~$\pi$. Based on
Lubiw's reductions~\cite{Lub81}, it is not hard to see that the problem is
\NP-complete when $k$ is not treated as parameter.

In this paper we consider the problem of finding isomorphisms with additional
constraints in the parameterized setting.
In our main result we formulate a \emph{graph isomorphism/automorphism problem
  with additional constraints} that generalizes Lubiw's setting as follows.
For a graph $X=(V,E)$, let $\pi\in\Aut(X)$ be an automorphism of~$X$.
We say that a permutation $\pi\in\Sym(V)$ \emph{satisfies} a formula~$F$ over
the variables in~$\Var(V)=\set{x_{u,v}}{u,v\in V}$ if $F$~is satisfied by the
assignment that has $x_{u,v}=1$ if and only if $u^\pi=v$. For example, the
conjunction $\bigwedge_{u\in V}\neg x_{uu}$ expresses the condition that $\pi$~is
fixed-point-free. We define:

\begin{problemdef}{$\ExactCNFGI$}
  Given two graphs $X_1=(V,E_1)$ and $X_2=(V,E_2)$, a CNF~formula~$F$ over
  $\Var(V)$, and $k\in\Nset$, decide whether there is an isomorphism from~$X_1$
  to~$X_2$ that has weight exactly~$k$ and satisfies~$F$. The parameter is
  $\card{F}+k$, where $\card{F}$~is the number of variables used in~$F$.
\end{problemdef}

In Section~\ref{sec:exactgi}, we first give an FPT algorithm for~$\ExactCNFGA$, the
automorphism version of this problem. The algorithm uses an orbit shrinking
technique that allows us to transform the input graph into a graph with bounded
color classes, preserving the existence of an exact weight~$k$ automorphism that
satisfies the formula~$F\!$. The bounded color class version is easy to solve
using color coding; see Section~\ref{Sec:bcc} for details. Building on this, we
show that $\ExactCNFGI$ is also in~\FPT. In particular, this allows us to
efficiently find isomorphisms of weight exactly~$k$, a problem left open
in~\cite{AKKT17}, and extends Schweitzer's result mentioned
above to the \emph{exact} case. In our earlier paper~\cite{AKKT17} we have shown
that the problem of \emph{exact weight~$k$ automorphism} is in~\FPT using a
simpler orbit shrinking technique which does not work for exact weight~$k$
isomorphisms. In this paper, we use some extra group-theoretic machinery to
obtain a more versatile orbit shrinking.

In Section~\ref{sec:fptgi}, we turn from restrictions on weight and complexity
to restrictions given only by a CNF~formula over Lubiw-style constraints. We
show that hypergraph isomorphism constrained by a CNF~formula is in
$\FPT^{\prob{GI}}$. Note that the problem remains \prob{GI}-hard even when the
formula is constantly true, so an FPT algorithm without \prob{GI}~oracle would
imply $\prob{GI}\in\cclass{P}$.

In Section~\ref{sec:exact-complexity}, we consider the problem of computing
graph isomorphisms of \emph{complexity} exactly~$t$: The complexity of a
permutation $\pi\in\Sym(V)$ is the minimum number of transpositions whose
product is~$\pi$. Checking for automorphisms or isomorphisms of complexity
exactly~$t$ is \NP-complete in the non-parameterized setting. We show that the
problem is in~\FPT with $t$~as parameter. Again, the ``at most~$t$'' version of
this problem was already shown to be in~\FPT by Schweitzer~\cite{Schw} as part
of his algorithmic strategy to solve the weight at most~$k$ problem. Our results
in Sections~\ref{sec:exactgi}~and~\ref{sec:exact-complexity} also hold for
hypergraphs when the maximum hyperedge size 
is taken as additional parameter.

In Section~\ref{sec:red-blue}, we examine a different restriction on the
automorphisms being searched for. Consider graphs $X=(V,E)$ with vertex set
partitioned into \RED and \BLUE. The \emph{Colored Graph Automorphism} problem
(defined in~\cite{DoFe}; we denote it \colGA), is to check if $X$~has an
automorphism that respects the partition and moves exactly~$k$ \BLUE vertices.
We showed in~\cite{AKKT17} that this problem is \Wone-hard. In our hardness
proof the orbits of the vertices in the \RED part of the graph have size at
most~4, while the ones for the \BLUE vertices have size~2. We show here that
this cannot be restricted any further. If we force the size of the orbits
of~$\Aut(X)$ in the \RED part to be bounded by~3 (i.e., the input graph has \RED
further partitioned into color classes of size at most~3 each), then the problem
to test whether there is an automorphism moving exactly~$k$ \BLUE vertices can
be solved in \FPT (with parameter~$k$). The \BLUE part of the graph remains
unconstrained. Observe that Schweitzer's problem~\cite{Schw} coincides with the
special case of this problem where there are no \RED vertices. This implies that
the non-parameterized version of \colGA is \NP-complete (even when $X$~has only
\BLUE vertices). Similarly, finding weight~$k$ automorphisms of a hypergraph
reduces to \colGA by taking the incidence graph, where the original vertices
become \BLUE and the vertices for hyperedges are \RED; note that this yields
another special case, where both \RED and \BLUE induce the empty graph,
respectively.

\section{Preliminaries}\label{sec:prelim}

We use standard permutation group terminology, see e.g.~\cite{DM96}.
Given a permutation $\sigma\in\Sym(V)$, its \emph{support} is
$\support(\sigma)=\set{u\in V}{u^\sigma\neq u}$ and its \emph{(Hamming) weight}
is $\card{\support(\sigma)}$. The \emph{complexity} of~$\sigma$ (sometimes
called its \emph{Cayley weight}) is the minimum number~$t$ such that $\sigma$
can be written as the product of $t$~transpositions.

Let $G\le\Sym(V)$ and $\pi\in\Sym(V)$; this includes the case $\pi=\id$. A
permutation $\sigma\in G\pi\setminus\set{\id}$ has \emph{minimal complexity
  in~$G\pi$} if for every way to express~$\sigma$ as the product of a minimum
number of transpositions $\sigma=\tau_1\dotsm\tau_{\compl(\sigma)}$ and every
$i\in\set{2,\dotsc,\compl(\sigma)}$ it holds that
$\tau_i\dotsm\tau_{\compl(\sigma)}\notin G\pi$. The following lemma observes
that every element of~$G\pi$ can be decomposed into minimal-complexity factors.
\begin{lemma}[{\cite[Lemma~2.2]{AKKT17}}]\label{lem:decompose-compl}
  Let~$G\pi$ be a coset of a permutation group~$G\!$ and let $\sigma\in
  G\pi\setminus\set{\id}$. Then for some $\ell\ge 1$ there are
  $\sigma_1,\dotsc,\sigma_{\ell-1}\in G\!$ with minimal complexity in~$G\!$ and
  $\sigma_\ell\in G\pi$ with minimal complexity in~$G\pi$ such that
  $\sigma=\sigma_1\dotsm\sigma_\ell$ and
  $\support(\sigma_i)\subseteq\support(\sigma)$ for each
  $i\in\set{1,\dotsc,\ell}$.
\end{lemma}

An action of a permutation group $G\le\Sym(V)$ on a set~$V'$ is a group
homomorphism $h\colon G\to\Sym(V')$; we denote the image of~$G$
under~$h$ by~$G(V')$. For $u\in V$, we denote its stabilizer
by~$G_u=\set{\pi\in G}{u^\pi=u}$. For $U\subseteq V$, we denote its pointwise
stabilizer by~$G_{[U]}=\set{\pi\in G}{\forall u\in U:u^\pi=u}$ and its setwise
stabilizer by~$G_{\set{U}}=\set{\pi\in G}{U^\pi=U}$. For
$S\subseteq\mathcal{P}(V)$, we let~$G_S=\set{\pi\in G}{\forall U\in S:
  U^\pi=U}$.

A \emph{hypergraph} $X=(V,E)$ consists of a vertex set~$V$ and a hyperedge
set~$E\subseteq\mathcal{P}(V)$. Graphs are the special case where $\card{e}=2$
for all $e\in E$. The \emph{degree} of a vertex~$v\in V$ is $\card{\set{e\in
    E}{v\in e}}$. A \emph{(vertex) coloring} of~$X$ is a partition of~$V$ into
color classes $\mathcal{C}=(C_1,\dotsc,C_m)$. The color classes~$\mathcal{C}$ are
\emph{$b$-bounded} if $\card{C_i}\le b$ for all $i\in[m]$. An \emph{isomorphism}
between two hypergraphs $X=(V,E)$~and~$X'=(V',E')$ (with color classes
$\mathcal{C}=(C_1,\dotsc,C_m)$~and~$\mathcal{C}'=(C'_1,\dotsc,C'_m)$) is a
bijection~$\pi\colon V\to V'$ such that $E'=\set[\big]{\set{\pi(v)}{v\in
    e}}{e\in E}$ (and $C'_i=\set[\big]{\pi(v)}{v\in C_i}$). The isomorphisms
from~$X$ to~$X'$ form a coset that we denote by $\Iso(X,X')$. The automorphisms
of a hypergraph~$X$ are the isomorphisms from~$X$ to itself; they form a group
which we denote by~$\Aut(X)$.

\section{Bounded color class size}\label{Sec:bcc}

To show that $\ExactCNFGA$ for hypergraphs with $b$-bounded color classes can be
solved in~\FPT, we recall our algorithm for exact weight~$k$ automorphisms of
bounded color class hypergraphs~\cite{AKKT17} and show how it can be adapted to
the additional constraints given by the input formula.

\begin{definition}\label{def:cc-min}
  Let $X=(V,E)$ be a hypergraph with color class set
  $\mathcal{C}=\set{C_1,\dotsc,C_m}$.
  \begin{enumerate}[nosep,label=(\alph*)]
   \item For a subset $\mathcal{C}'\subseteq\mathcal{C}$, we say that a
    color-preserving permutation $\pi\in\Sym(V)$
    \emph{$\mathcal{C}'$-satisfies} a CNF~formula~$F$ over~$\Var(V)$
    if every clause of~$F$ contains a literal $x_{u,v}$~or~$\lnot x_{u,v}$ with
    $u\in\bigcup\mathcal{C}'$ that is satisfied by~$\pi$.
   \item For a color-preserving permutation $\pi\in\Sym(V)$, let
    $\mathcal{C}[\pi] = \set{C_i\in\mathcal{C}}{ \exists v\in C_i: v^\pi\ne v}$
    be the subset of color classes that intersect~$\support(\pi)$. For a subset
    $\mathcal{C}'\subseteq \mathcal{C}[\pi]$, we define the permutation
    $\pi_{\mathcal{C}'}\in\Sym(V)$ as
    \[\pi_{\mathcal{C}'}(v)=\begin{cases}
        v^\pi,& \text{if }v\in\bigcup\mathcal{C}',\\
        v,& \text{if }v\not\in\bigcup\mathcal{C}'.
      \end{cases}\]
    Note that $\pi_{\mathcal{C}[\pi]}=\pi$.
   \item A color-preserving automorphism $\sigma\ne\id$ of~$X$ is said to be
    \emph{color-class-minimal}, if for every set~$\mathcal{C}'$ with
    $\emptyset\subsetneq\mathcal{C}'\subsetneq\mathcal{C}[\sigma]$, the
    permutation~$\sigma_{\mathcal{C}'}$ is not in~$\Aut(X)$.
  \end{enumerate}
\end{definition}

\begin{lemma}\label{lem:cc-min-decompose}
  Let $X=(V,E)$ be a hypergraph with color class set
  $\mathcal{C}=\{C_1,C_2,\ldots,C_m\}$. For
  $\emptyset\ne\mathcal{C}'\subseteq\mathcal{C}$ and a CNF~formula~$F\!$
  over~$\Var(V)$,  the following statements are equivalent:
  \begin{itemize}[nosep]
   \item There is a nontrivial automorphism~$\sigma$ of~$X\!$ with
    $\mathcal{C}[\sigma]=\mathcal{C}'\!$ that satisfies~$F\!$.
   \item $\mathcal{C}'\!$ can be partitioned into
    $\mathcal{C}_1,\dotsc,\mathcal{C}_\ell$ and $F\!$~(seen as a set of clauses) can be
    partitioned into CNF formulas $F_0,\dotsc,F_\ell$ such that $F_0$~is
    $(\mathcal{C}\setminus\mathcal{C}')$-satisfied by~$\id$ and for each
    $i\in\set{1,\dotsc,\ell}$ there is a color-class-minimal
    automorphism~$\sigma_i$ of~$X\!$ with $\mathcal{C}[\sigma_i]=\mathcal{C}_i$
    that $\mathcal{C}_i$-satisfies~$F_i$.
  \end{itemize}
  Moreover, the automorphisms $\sigma$~and~$\sigma_i$ can be chosen to satisfy
  $\sigma_i=\sigma_{\mathcal{C}_i}$ for $1\le i\le\ell$, respectively.
\end{lemma}

\begin{proof}
  To show the forward direction, let~$\sigma$ be a nontrivial
  automorphism of~$X$ with $\mathcal{C}[\sigma]=\mathcal{C}'$ that
  satisfies~$F\!$. We put those clauses of~$F$ into~$F_0$ that are
  $(\mathcal{C}\setminus\mathcal{C}')$-satisfied by~$\id$, and the
  remaining clauses of~$F$ into the CNF~formula~$F'$. Note that
  $F'$~is $\mathcal{C}'$-satisfied by~$\sigma$: Every clause of~$F$
  must contain a literal $x_{u,v}$~or~$\lnot x_{u,v}$ that is satisfied
  by~$\sigma$. If this literal is not $\mathcal{C}'$-satisfied
  by~$\sigma$, we have $u\notin\bigcup\mathcal{C}'$ and thus
  $u^\sigma=u$, so this clause is contained in~$F_0$.

  We show by induction on~$\card[\big]{\mathcal{C}[\sigma]}$ that for
  an automorphism~$\sigma$ of~$X$ which
  $\mathcal{C}[\sigma]$-satisfies a CNF~formula~$F'$ over~$\Var(V)$,
  we can partition $\mathcal{C}[\sigma]$ into
  $\mathcal{C}_1,\dotsc,\mathcal{C}_\ell$ and the clauses of~$F'$ into
  $F_1,\dotsc,F_\ell$ such that $\sigma_{\mathcal{C}_i}$~is a
  color-class-minimal automorphism of~$X\!$ that
  $\mathcal{C}_i$-satisfies~$F_i$, for $1\le i\le\ell$. If
  $\sigma$~itself is color-class-minimal, which always happens if
  $\card[\big]{\mathcal{C}[\sigma]}=1$, we are done: We can set
  $\ell=1$, $\mathcal{C}_1=\mathcal{C}[\sigma]$, and
  $F_1=F'$. Otherwise there is a non-empty
  $\mathcal{D}\subsetneq\mathcal{C}[\sigma]$ such that
  $\phi=\sigma_{\mathcal{D}}\in\Aut(X)$. This implies
  $\phi'=\sigma\phi^{-1}\in\Aut(X)$. Note that
  $\phi'=\sigma_{\mathcal{C}[\sigma]\setminus\mathcal{D}}$ and thus
  $\mathcal{C}[\phi']=\mathcal{C}[\sigma]\setminus\mathcal{D}$. Next,
  we partition the clauses of~$F'$ into two CNF~formulas
  $\hat{F}$~and~$\hat{F}'$: If a clause of~$F'\!$ is
  $\mathcal{D}$-satisfied by~$\sigma$, we include it in~$\hat{F}$;
  otherwise we include it in~$\hat{F}'\!$. In the former case, this
  implies that this clause is also $\mathcal{D}$-satisfied by~$\phi$,
  as $u^\phi=u^\sigma$ for all $u\in\bigcup\mathcal{D}$. In the latter
  case, the clause must be
  $(\mathcal{C}[\sigma]\setminus\mathcal{D})$-satisfied by~$\sigma$,
  and consequently also by~$\phi'$, as $u^\sigma=u^{\phi'}$ for all
  $u\in\bigcup(\mathcal{C}[\sigma]\setminus\mathcal{D})$. Thus
  $\hat{F}$ is $\mathcal{C}[\phi]$-satisfied by~$\phi$, and $\hat{F}'$
  is $\mathcal{C}[\phi']$-satisfied by~$\phi'$; so we can apply the
  inductive hypothesis to both $\phi$~and~$\phi'$. This yields
  partitions of $\mathcal{D}$~and~$\hat{F}$ as well as
  $\mathcal{C}[\sigma]\setminus\mathcal{D}$~and~$\hat{F}'$, which we
  can combine to obtain the desired partitions of
  $\mathcal{C}[\sigma]$~and~$F'$.

  To show the backward direction, let
  $\mathcal{C}_1,\dotsc,\mathcal{C}_\ell$ be a partition
  of~$\mathcal{C}'$, let $F_0,\dotsc,F_\ell$ be a partition of the
  clauses of~$F$, and let $\sigma_1,\dotsc,\sigma_\ell$ be
  color-preserving automorphisms of~$X\!$ with
  $\mathcal{C}[\sigma_i]=\mathcal{C}_i$ such that $F_i$~is
  $\mathcal{C}_i$-satisfied by~$\sigma_i$ for $1\le i\le\ell$, and
  $F_0$~is $(\mathcal{C}\setminus\mathcal{C}')$-satisfied
  by~$\id$. Consider the automorphism
  $\sigma=\sigma_1\dotsm\sigma_\ell$. As
  $\mathcal{C}_i\cap\mathcal{C}_j=\emptyset$ for $i\neq j$, the
  following definition of~$\sigma$ is equivalent and well-defined:
  \[v^\sigma=\begin{cases}
    v^{\sigma_i}&\text{if }\exists i\in\set{1,\dotsc,\ell}: v\in\bigcup\mathcal{C}_i\\
    v&\text{otherwise}
  \end{cases}\] Thus we have $\sigma_i=\sigma_{\mathcal{C}_i}$ and
  $\mathcal{C}[\sigma]=\mathcal{C}'$. Moreover, any clause of~$F$ is
  contained in some~$F_i$. If $i>0$, this clause is
  $\mathcal{C}_i$-satisfied by~$\sigma_i$ and thus also by~$\sigma$,
  as $u^{\sigma_i}=u^\sigma$ for all $u\in\bigcup\mathcal{C}_i$. It
  remains to consider the case $i=0$. Then the clause is
  $(\mathcal{C}\setminus\mathcal{C}')$-satisfied by~$\id$ and thus
  also by~$\sigma$, as $u^\sigma=u$ for all
  $u\in\bigcup(\mathcal{C}\setminus\mathcal{C}')$.
\end{proof}

In~\cite{AKKT17} an algorithm is presented that, when given a
hypergraph~$X$ on vertex set~$V$ with $b$-bounded color classes and $k\in\Nset$, computes
all color-class-minimal automorphisms of~$X$ that have weight
exactly~$k$ in
$\mathcal{O}\paren[\big]{(kb!)^{\mathcal{O}(k^2)}\poly(N)}$~time. We
use it as a building block for the following algorithm (see line~\ref{line:weight}).

\pagebreak

\begin{lstlisting}[caption={$\algo{ColorExactCNFGA}_b(X,\mathcal{C},k,F)$},label=alg:colorexactcnfga]
Input: A hypergraph $X=(V,E)$ with (*$b$-bounded*) color classes $\mathcal{C}=\set{C_1,\dotsc,C_m}$, (*a~parameter*) $k\in\Nset$, (*and*) a CNF formula $F$ over $\Var(V)$
Output: A (*color-preserving*) automorphism $\sigma$ of $X$ with $\card{\support(\sigma)}=k$ that satisfies $F\!$, (*or~$\bot$~if*) none exists
$A_0=\set{\id}$
for $i\in\set{1,\dotsc,k}$ do
    $A_i\assign\set{\sigma\in\Aut(X)}{\sigma\text{ is color-class-minimal and has weight }i}\label{line:weight}$ // see~\cite{AKKT17}
for $h\in\mathcal{H}_{\mathcal{C},k}$ do // $\mathcal{H}_{\mathcal{C},k}$ is the perfect family of hash functions $h\colon\mathcal{C}\to [k]$ from~\cite{FKS84}
    for $\ell\in\set{1,\dotsc,k}, h'\colon[k]\to[\ell]$ do
        for $(k_1,\dotsc,k_\ell)\in\set{0,\dotsc,k}\relax\cramped{{}^\ell}$ with $\cramped\sum_{i=1}^\ell k_i=k$ do
            for each partition of the clauses of $F$ into $F_0,\dotsc,F_\ell$ do
                if $\forall i\in\set{1,\dotsc,\ell}:\exists \sigma_i\in A_{k_i}:\support(\sigma_i)\subseteq \bigcup (h'\circ h)^{-1}(i)$ and $F_i$ is (*\mbox{$\mathcal{C}[\sigma_i]$-satisfied}*) by $\sigma_i$, and $F_0$ is (*$(\mathcal{C}\setminus\cramped\bigcup_{i=1}^\ell\mathcal{C}[\sigma_i])$-satisfied*) by $\id\label{line:cccond}$ then
                    return $\sigma=\sigma_1\dotsm\sigma_\ell$
return $\bot$
\end{lstlisting}
\begin{theorem}\label{thm:colorcnfga}
  Given
  a hypergraph $X=(V,E)$ with $b$-bounded color classes~$\mathcal{C}$, a
  CNF~formula $F\!$ over~$\Var(V)$, and $k\in\Nset$, the algorithm
  $\algo{ColorExactCNFGA}_b(X,\mathcal{C},k,F)$ computes a color-preserving
  automorphism~$\sigma$ of~$X\!$ with weight~$k$ that satisfies~$F\!$ in
  $(kb!)^{\mathcal{O}(k^2)}k^{\mathcal{O}(\card{F})}\poly(N)$ time (where
  $N\!$~is the size of~$X\!$), or determines that none exists.
\end{theorem}
\begin{proof}
  If the algorithm returns $\sigma=\sigma_1\dotsm\sigma_\ell$, we know
  $\sigma_i\in A_{k_i}$ and $\support(\sigma_i)\subseteq \bigcup (h'\circ
  h)^{-1}(i)$. As these sets are disjoint, we have
  $\card{\support(\sigma)}=\sum_{i=1}^\ell\card{\support(\sigma_i)}=k$, and
  Lemma~\ref{lem:cc-min-decompose} implies that~$\sigma$ satisfies~$F\!$.

  We next show that the algorithm does not return~$\bot$ if there is an
  automorphism~$\pi$ of~$X$ that has weight~$k$ and satisfies~$F\!$. By
  Lemma~\ref{lem:cc-min-decompose}, we can partition $\mathcal{C}[\pi]$ into
  $\mathcal{C}_1,\dotsc,\mathcal{C}_\ell$ and the clauses of~$F$ into
  $F_0\dotsc,F_\ell$ such that $F_0$ is
  $(\mathcal{C}\setminus\mathcal{C}[\pi])$-satisfied by~$\id$ and, for $1\le i\le\ell$, the
  permutation~$\pi_i=\pi_{\mathcal{C}_i}$ is a color-class-minimal
  automorphism of~$X\!$ that $\mathcal{C}[\pi_i]$-satisfies~$F$. Now consider
  the iteration of the loop where $h$~is injective on~$\mathcal{C}[\pi]$;
  such an~$h$ must exist as it is chosen from a perfect hash family. Now let
  $h'\colon[k]\to[\ell]$ be a function with $h'\paren[\big]{h(C)}=i$ if
  $C\in\mathcal{C}[\pi_i]$; such an~$h'$ exists because $h$~is injective
  on~$\mathcal{C}[\pi]$. In the loop iterations where $h'$~and the partition
  of~$F$ into $F_0\dotsc,F_\ell$ is considered, the condition on
  line~\ref{line:cccond} is true (at least) with $\sigma_i=\pi_i$, so the
  algorithm does not return~$\bot$.

  Line~\ref{line:weight} can be implemented by using the
  algorithm~$\algo{ColoredAut}_{k,b}(X)$ from~\cite{AKKT17} which runs in
  $\mathcal{O}\paren[\big]{(kb!)^{\mathcal{O}(k^2)}\poly(N)}$~time, and this
  also bounds~$\card{A_i}$. As $\card{\mathcal{C}}\le n$, the perfect hash
  family~$\mathcal{H}_{\mathcal{C},k}$ has size $2^{\mathcal{O}(k)}\log^2n$, and
  can also be computed in this time. The inner
  loops take at most $k^k$, $k^k$ and $(k+1)^{\card{F}}$ iterations,
  respectively. Together, this yields a runtime of
  $(kb!)^{\mathcal{O}(k^2)}k^{\mathcal{O}(\card{F})}\poly(N)$.
\end{proof}

\section{Exact weight}\label{sec:exactgi}

In this section, we show that finding isomorphisms that have an exactly
prescribed weight and satisfy a CNF formula is fixed parameter tractable. In
fact, we show that this is true even for hypergraphs, when the maximum hyperedge
size~$d$ is taken as additional parameter.

\begin{problemdef}{$\ExactCNFHGI$}
  Given two hypergraphs $X_1=(V,E_1)$ and $X_2=(V,E_2)$ with hyperedge size
  bounded by~$d$, a CNF~formula~$F$ over $\Var(V)$, and $k\in\Nset$, decide
  whether there is an isomorphism from~$X_1$ to~$X_2$ of weight~$k$ that
  satisfies~$F$. The parameter is $\card{F}+k+d$.
\end{problemdef}

Our approach is to reduce $\ExactCNFHGI$ to $\ExactCNFHGA$ (the analogous problem for
automorphisms), which we solve first.

We require some permutation group theory definitions. Let $G\le\Sym(V)$ be
a permutation group. The group~$G$ partitions~$V$ into orbits:
$V=\Omega_1\cup\Omega_2\cdots\cup\Omega_r$. On each orbit~$\Omega_i$, the
group~$G$ acts transitively.
A subset $\Delta\subseteq \Omega_i$ is a \emph{block} of the group~$G$ if for
all $\pi\in G$ either $\Delta^\pi=\Delta$ or $\Delta^\pi\cap \Delta=\emptyset$.
Clearly, $\Omega_i$ is itself a block, and so are all singleton sets. These
are \emph{trivial} blocks. Other blocks are \emph{nontrivial}. If $G$ has no
nontrivial blocks it is \emph{primitive}. If $G$ is not primitive, we can
partition~$\Omega_i$ into blocks $\Omega_i = \Delta_1\cup \Delta_2\cup \cdots
\cup \Delta_s$, where each~$\Delta_j$ is a \emph{maximal nontrivial block}. Then
the group~$G$ acts primitively on the block system
$\{\Delta_1,\Delta_2,\ldots,\Delta_s\}$. In this action, a permutation $\pi\in
G$ maps~$\Delta_i$ to $\Delta_i^\pi=\set{u^\pi}{u\in\Delta_i}$.

The following two theorems are the main group-theoretic ingredients to our
algorithms; they imply that every primitive group on a sufficiently large set
contains the alternating group.

\begin{theorem}[{\cite[Theorem 3.3A]{DM96}}]\label{thm:primitive-small-support}
  Suppose $G\le \Sym(V)$ is a primitive subgroup of $\Sym(V)$. If $G$ contains
  an element~$\pi$ such that $\card{\support(\pi)}=3$ then $G$~contains the
  alternating group~$\Alt(V)$. If $G$ contains an element~$\pi$ such that
  $\card{\support(\pi)}=2$ then $G=\Sym(V)$.
\end{theorem}

\begin{theorem}[{\cite[Theorem 3.3D]{DM96}}]\label{thm:primitive-large-support}
  If $G\le \Sym(V)$ is primitive with $G\notin\set{\Alt(V),\Sym(V)}$ and
  contains an element~$\pi$ such that $\card{\support(\pi)}=m$ (for some $m\ge
  4$) then $\card{V}\le (m-1)^{2m}$.
\end{theorem}

The following lemma implies that the alternating group in a large orbit survives
fixing vertices in a smaller orbit.

\begin{lemma}\label{lem:bigger-survives}
  Let $G\le \Sym(\Omega_1\cup \Omega_2)$ be a permutation group such that
  $\Omega_1$ is an orbit of $G$, and $\card{\Omega_1}\ge 5$. Recall that
  $G(\Omega_i)$ denotes the image of~$G$ under its action on~$\Omega_i$. Suppose
  $G(\Omega_1)\in\set{\Alt(\Omega_1),\Sym(\Omega_1)}$ and
  $\card{G(\Omega_1)}>\card{G(\Omega_2)}$. Then for some subgroup $H$ of
  $G(\Omega_2)$, the group~$G$ contains the product group $\Alt(\Omega_1)\times
  H$. In particular, the pointwise stabilizer~$G_{[\Omega_2]}$ contains the
  subgroup $\Alt(\Omega_1)\times\set{\id}$.
\end{lemma}
\begin{proof}
  Let $p_2\colon G\to G(\Omega_2)$ denote the surjective projection
  homomorphism. Then
  \[
    \Ker(p_2) = \set[\big]{(x,1)}{(x,1)\in G}.
  \]
  Let $K=\set[\big]{x}{(x,1)\in\Ker(p_2)}$. It is easily checked that $K$ is a
  normal subgroup of~$G(\Omega_1)$. As
  $G(\Omega_1)\in\set{\Alt(\Omega_1),\Sym(\Omega_1)}$, $\Alt(\Omega_1)$ is
  simple, and the only nontrivial normal subgroup of $\Sym(\Omega_1)$ is
  $\Alt(\Omega_1)$, it follows that either $K=\{1\}$ or $\Alt(\Omega_1)\le K$.

  \begin{description}
   \item[Case 1.] Suppose $K=\{1\}$. In this case $\Ker(p_2)$ is trivial. Thus,
    $p_2$ is an isomorphism from~$G$ to~$G(\Omega_2)$ implying that
    $\card{G}=\card{G(\Omega_2)}$. As $\card{G}\ge\card{G(\Omega_1)}$, this
    contradicts the assumption that $\card{G(\Omega_2)}<\card{G(\Omega_1)}$.
   \item[Case 2.] Suppose $\Alt(\Omega_1)\le K$. Consider the other surjective
    projection homomorphism $p_1\colon G\to G(\Omega_1)$. Then
    $\Ker(p_1) = \set[\big]{(1,y)}{(1,y)\in G}$,
    and $H=\set[\big]{y}{(1,y)\in\Ker(p_1)}$ is a normal subgroup
    of~$G(\Omega_2)$. We show that $G$~contains the product group
    $\Alt(\Omega_1)\times H$ as claimed by the lemma.

    Consider any pair $(x,y)\in\Alt(\Omega_1)\times H$. We can write it as $
    (x,y)=(x,1)\cdot (1,y), $ and note that by definition $(x,1)\in\Ker(p_2)$
    and $(1,y)\in\Ker(p_1)$. As both $\Ker(p_1)$ and $\Ker(p_2)$ are subgroups
    of $G$, it follows that $(x,y)\in G$.
    \qedhere
  \end{description}
\end{proof}

\begin{remark}
  In a special case of Lemma~\ref{lem:bigger-survives}, suppose $G\le
  \Sym(\Omega_1\cup\Omega_2)$ such that $\Omega_1$ is an orbit of~$G$,
  $\card{\Omega_1}\ge \max\set[\big]{5,\card{\Omega_2}+1}$, and
  $\Alt(\Omega_1)\le G(\Omega_1)$. As $\card{G(\Omega_1)}>\card{G(\Omega_2)}$ is
  implied by this assumption, the consequence of the lemma follows.
\end{remark}

The effect of fixing vertices of some orbit on other orbits of the same size
depends on how the group relates these orbits to each other.

\begin{definition}\label{def:linked}
  Two orbits $\Omega_1$~and~$\Omega_2$ of a permutation group $G\le\Sym(V)$
  are \emph{linked} if there is a group isomorphism $\sigma\colon G(\Omega_1)\to
  G(\Omega_2)$ with
  $G(\Omega_1\cup\Omega_2)=\set[\big]{(\phi,\sigma(\phi))}{\phi\in
    G(\Omega_1)}$. (This happens if and only if both
  $G(\Omega_1)$~and~$G(\Omega_2)$ are isomorphic to $G(\Omega_1\cup\Omega_2)$.)
\end{definition}

We next show that two large orbits where the group action includes the
alternating group are (nearly) independent unless they are linked.

\begin{lemma}\label{lem:linked-or-alt-product}
  Suppose $G\le\Sym(V)$ where $V=\Omega_1\cup\Omega_2$ is its orbit
  partition such that $\card{\Omega_i}\ge 5$ and
  $G(\Omega_i)\in\set{\Alt(\Omega_i),\Sym(\Omega_i)}$ for $i=1,2$. Then either
  $\Omega_1$~and~$\Omega_2$ are linked in~$G$, or
  $G$~contains $\Alt(\Omega_1)\times\Alt(\Omega_2)$.
\end{lemma}
\begin{proof}
  For $i=1,2$, let $p_i\colon G\to G(\Omega_i)$ denote the surjective projection
  homomorphisms. Further, let $H=\set[\big]{x}{(1,x)\in\Ker(p_1)}$ and
  $K=\set[\big]{x}{(x,1)\in\Ker(p_2)}$. It is easily checked that $H$ is a
  normal subgroup of~$G(\Omega_2)$. Therefore, $H$ is either $G(\Omega_2)$ or
  $\set{\id}$ or $\Alt(\Omega_2)$ (note: the last case coincides with the first
  if $G(\Omega_2)=\Alt(\Omega_2)$). Similarly, $K$ is a normal subgroup
  of~$G(\Omega_1)$ and thus either $G(\Omega_1)$ or $\set{\id}$ or
  $\Alt(\Omega_1)$.
  \begin{description}
   \item[Case 1: $H=\set{\id}$ (the case $K=\set{\id}$ is symmetric).] Then
    $\Ker(p_1)$ is trivial, and $p_1$ is an isomorphism from~$G$
    to~$G(\Omega_1)$, implying that $\card{G}=\card{G(\Omega_1)}$. By the basic
    isomorphism theorem, we have $\frac{G}{\Ker(p_2)}\cong G(\Omega_2)$, and
    hence $\card{G}=\card{\Ker(p_2)}\cdot\card{G(\Omega_2)}$. But $G$ is
    isomorphic to~$G(\Omega_1)$ and hence has only three possible normal
    subgroups: isomorphic to~$G(\Omega_1)$, isomorphic to $\Alt(\Omega_1)$, or
    isomorphic to $\{\id\}$. In the first two cases,
    $\card{\Ker(p_2)}\ge\card{\Alt(\Omega_2)}\ge \card{G(\Omega_1)}/2$. Hence,
    $\card{G}>\card{G(\Omega_1)}$; a contradiction. Thus, $\Ker(p_2)=\{\id\}$,
    which implies that $\Omega_1$~and~$\Omega_2$ are linked in~$G$.
   \item[Case 2: $H=G(\Omega_2)$ (the case $K=G(\Omega_1)$ is symmetric).]
    Consider any pair $(y,x)\in G(\Omega_1)\times G(\Omega_2)$. Since $y\in
    G(\Omega_1)$ there is a $z\in G(\Omega_2)$ such that $(y,z)\in G$. Now, we
    can write $(y,x)=(y,z)(1,z^{-1}x)$, and note that $(y,z)\in G$ and
    $(1,z^{-1}x)\in\Ker(p_1)\subseteq G$ by assumption on $H$. Therefore,
    $(y,x)\in G$ implying that $G=G(\Omega_1)\times G(\Omega_2)$, and thus
    $\Alt(\Omega_1)\times\Alt(\Omega_2)\le G$.
   \item[Case 3.]  Finally, we are left with the possibility that
    $G(\Omega_1)=\Sym(\Omega_1)$, $G(\Omega_2)=\Sym(\Omega_2)$,
    $H=\Alt(\Omega_2)$ and $K=\Alt(\Omega_1)$. In this case $G$~contains
    $\Alt(\Omega_1)\times \Alt(\Omega_2)$.
    \qedhere
  \end{description}
\end{proof}

The last ingredient for our algorithm is that when there are two linked orbits
where the group action includes the alternating group, fixing a vertex in one
orbit is equivalent to fixing some vertex of the other orbit.

\begin{lemma}[{\cite[Theorem 5.2A]{DM96}}]%
\label{lem:altsub}
  Let $n=\card{V}> 9$. Suppose $G$ is a subgroup of~$\Alt(V)$ of
  index strictly less than $\binom{n}{2}$. Then, for some point
  $u\in V$, the group~$G$ is the pointwise stabilizer subgroup
  $\Alt(V)_u$.
\end{lemma}

\begin{corollary}\label{cor:linked}
  Let $\Omega_1$~and~$\Omega_2$ be two linked orbits of a permutation group
  $G\le\Sym(V)$ with $\Alt(\Omega_1)\le G(\Omega_1)$ and
  $\card{\Omega_1}=\card{\Omega_2}>9$.
  Then for each $u\in\Omega_1$ there is a $v\in\Omega_2$
  such that $G_u=G_v$.
\end{corollary}
\begin{proof}
  Let~$\sigma\colon G(\Omega_1)\to G(\Omega_2)$ be the group isomorphism which
  witnesses that $\Omega_1$~and~$\Omega_2$ are linked. As
  $G(\Omega_1)\in\set[\big]{\Alt(\Omega_1),\Sym(\Omega_1)}$, the index
  of~$G_u(\Omega_1)$ in~$G(\Omega_1)$ is $n!/(n-1)!=n$. As $\sigma$~is a group
  isomorphism, the index of~$\sigma\paren[\big]{G_u(\Omega_1)}$
  in~$\sigma\paren[\big]{G(\Omega_1)}=G(\Omega_2)$ is also~$n$. Thus
  Lemma~\ref{lem:altsub} implies that there is $v\in\Omega_2$ such that
  $G_v(\Omega_2)=\sigma\paren[\big]{G_u(\Omega_1)}$. As this implies
  $\sigma^{-1}\paren[\big]{G_v(\Omega_2)}=G_u(\Omega_1)$, it follows that
  $G_u=G_v$.
\end{proof}

\begin{lstlisting}[caption={$\algo{Exact-CNF-HGA}_d(X,k,F)$},label=alg:exactcnfhga]
Input: (*A hypergraph $X$ with hyperedge size bounded by $d$, a parameter $k$ and a formula $F$*)
Output: (*An automorphism~$\sigma$ of~$X$ with $\card{\support(\sigma)}=k$ that satisfies $F$, or $\bot$ if none exists*)
$T \assign{}$ the vertices of $X$ that are mentioned in $F$
$G\assign\angle[\big]{\set[\big]{\sigma\in\Aut(X)}{\sigma\text{ has minimal complexity in }\Aut(X)\text{ and }\card{\support(\sigma)}\le k}}$ //~see~\mbox{\cite[Algorithm~3]{AKKT17}}\label{line:mincomplautos}
while $G$ contains an orbit of size more than $\frac{k}{2}\cdot{}\max\set{(k-1)^{2k},\card{T}+k,9}$ do
    repeat
        $\mathcal{O}\assign{}$(*the set of all $G$-orbits*)
        for $\Omega\in\mathcal{O}$ do
            $\blocks(\Omega)\assign{}$a maximal block system of $\Omega$ in $G$
            if $\exists \Delta\in\blocks(\Omega):\card{\Delta}>\frac{k}{2}$ or $\card{\blocks(\Omega)}>(k-1)^{2k}\land\Alt\paren[\big]{\blocks(\Omega)}\nleq G\paren[\big]{\blocks(\Omega)}$ then
                $G\assign G_{\blocks(\Omega)}$ // the setwise stabilizer of all $\Delta\in\blocks(\Omega)$ \label{line:shrink1}
    until $G$ remains unchanged
    choose $\Omega_{\max}\in\mathcal{O}$ such that $\card{\blocks(\Omega_{\max})}\ge\card{\blocks(\Omega)}$ (*for*) all $\Omega\in\mathcal{O}$
    if $\card{\blocks(\Omega_{\max})}>\max\set{(k-1)^{2k},\card{T}+k,9}\label{line:large-orbit}$ then
        $H\assign G_{[T]}$ // the pointwise stabilizer of $T$
        $\Omega_H\assign{}$(*the largest $H$-orbit that is contained in $\Omega_{\max}$*)
        $\mathcal{B}_H\assign\set[\big]{\Delta\in\blocks(\Omega_{\max})}{\Delta\subseteq \Omega_H}$
        choose $\Delta\in\mathcal{B}_H\label{line:delta}$
        $G\assign G_{\set{\Delta}}$ // the setwise stabilizer of~$\Delta$\label{line:shrink2}
$b\assign\frac{k}{2}\cdot{}\max\set{(k-1)^{2k},\card{T}+k,9}$
$\mathcal{O}\assign{}$(*the set of all $G$-orbits*)
return $\algo{ColorExactCNFGA}_{b}(X,\mathcal{O},k,F)$ // see Algorithm \ref{alg:colorexactcnfga}\label{line:bcc}
\end{lstlisting}

\begin{theorem}\label{thm:cnfhga}
  Algorithm~\ref{alg:exactcnfhga} solves $\ExactCNFHGA$ in time
  $\paren[\big]{d(k^k+\card{F})!}^{\mathcal{O}(k^2)}\poly(N)$.
\end{theorem}
\begin{proof}
  Suppose there is some $\pi\in\Aut(X)$ of weight exactly~$k$ that
  satisfies~$F$.

  By Lemma~\ref{lem:decompose-compl}, the
  automorphism~$\pi$ can be decomposed as a product of
  minimal-complexity automorphisms of weight at most~$k$, which
  implies $\pi\in G$.

  We will show that whenever the algorithm shrinks~$G$, some weight~$k$
  automorphism of~$X$ that satisfies~$F$ survives. For the shrinking in
  line~\ref{line:shrink1} we need to consider two cases. If $\Omega$ is an orbit with
  $\card{\Delta}>k/2$ for some (and thus all) $\Delta\in\blocks(\Omega)$, then none
  of these blocks are moved by~$\pi$. Indeed, if $\pi$ would move one block, it
  would have to move at least one further block, contradicting
  $\card{\support(\pi)}=k$. On the other hand, if $\card{\blocks(\Omega)}>(k-1)^{2k}$
  and $G\paren[\big]{\blocks(\Omega)}$ does not contain the alternating group, then
  Theorems~\ref{thm:primitive-small-support}~and~\ref{thm:primitive-large-support} imply that the primitive
  group~$G\paren[\big]{\blocks(\Omega)}$ contains no nontrivial element that moves at most~$k$
  elements of~$\blocks(\Omega)$. In particular, $\pi$~setwise stabilizes all
  $\Delta\in\blocks(\Omega)$ and thus survives the shrinking.

  We now turn to the other shrinking of~$G$, which occurs in
  line~\ref{line:shrink2}. Note that this can only happen if
  $\card{\blocks(\Omega_{\max})}>(k-1)^{2k}$ because of the if-condition on
  line~\ref{line:large-orbit}. This implies, as the last execution of the
  repeat-loop resulted in no further shrinking of~$G$, that
  $\Alt\paren[\big]{\blocks(\Omega_{\max})}\le G\paren[\big]{\blocks(\Omega_{\max})}$. Let
  $\mathcal{T}=\bigcup_{\Omega\in\mathcal{O}}\set{\Delta\in\blocks(\Omega)}{\Delta\cap
    T\neq\emptyset}$ be the set of all blocks with vertices from~$T$ and let
  $R=G_{\mathcal{T}}$ be the setwise stabilizer of these blocks. Note that
  $H\le R\le G$. We next show that a sufficiently large part
  of~$\Alt\paren[\big]{\blocks(\Omega_{\max})}$ survives in~$R$.
  \begin{claim}
    Let $\Omega_{R}$ be the largest orbit of~$R$ that is contained in~$\Omega_{\max}$.
    Then the set
    $\mathcal{B}_{R}=\set[\big]{\Delta\in\blocks(\Omega_{\max})}{\Delta\subseteq
      \Omega_{R}}$~is a maximal block system for the orbit~$\Omega_{R}$ of~$R$.
    Moreover, $\card{\mathcal{B}_{R}}>k$ and $\Alt(\mathcal{B}_{R})\le
    R(\mathcal{B}_{R})$.
  \end{claim}
  Let $\mathcal{O}=\set{\Omega_1,\dotsc,\Omega_k,\dotsc,\Omega_\ell}$ be an enumeration of the
  orbits of~$G$ such that $\Omega_i$~is linked to~$\Omega_{\max}$ if and only if $i>k$.
  Consider the sequence of subgroups $R=R^{(\ell)}\le\dotsm\le R^{(0)}=G$, where
  $R^{(i)}$~is the subgroup of~$R^{(i-1)}$ that setwise stabilizes all
  $\Delta\in\mathcal{T}$ with $\Delta\subseteq \Omega_i$. As first step, we
  inductively show for $i\le k$ that all orbits linked to~$\Omega_{\max}$ in~$G$
  (including itself) remain orbits of~$R^{(i)}$, and that
  $\Alt\paren[\big]{\blocks(\Omega)}\le R^{(i)}\paren[\big]{\blocks(\Omega)}$ for all
  $\Omega\in\mathcal{O}$ with $\card{\blocks(\Omega)}=\card{\blocks(\Omega_{\max})}$ that are
  still an orbit of~$R^{(i)}$.
  \begin{description}
   \item[Case 1:] Suppose that $\card{\blocks(\Omega_{\max})}>\card{\blocks(\Omega_i)}$.
    Consider any~$\Omega\in\mathcal{O}$ with
    $\card{\blocks(\Omega)}=\card{\blocks(\Omega_{\max})}$ that is an orbit
    of~$R^{(i-1)}$; this includes all orbits linked to~$\Omega_{\max}$ in~$G$. By the
    induction hypothesis, we know $\Alt\paren[\big]{\blocks(\Omega)}\le
    R^{(i-1)}\paren[\big]{\blocks(\Omega)}$, which implies
    $\card[\big]{R^{(i-1)}\paren[\big]{\blocks(\Omega)}}>
    \card[\big]{R^{(i-1)}\paren[\big]{\blocks(\Omega_i)}}$. Thus we can apply
    Lemma~\ref{lem:bigger-survives} to $R^{(i-1)}\paren[\big]{\blocks(\Omega)\times\blocks(\Omega_i)}$.
    This gives us $\Alt\paren[\big]{\blocks(\Omega)}\times\set{\id}\le
    R^{(i-1)}\paren[\big]{\blocks(\Omega)\times\blocks(\Omega_i)}$. Thus fixing some
    blocks of~$\Omega_i$ in~$R^{(i-1)}$ preserves the alternating group
    $\Alt\paren[\big]{\blocks(\Omega)}$ in~$R^{(i)}\paren[\big]{\blocks(\Omega)}$. In
    particular, $\Omega$~is also an orbit of~$R^{(i)}$.
   \item[Case 2:] Suppose that $\card{\blocks(\Omega_{\max})}=\card{\blocks(\Omega_i)}$
    and that $\Omega_i$~is no longer an orbit of~$R^{(i-1)}$. This again implies
    $\card[\big]{R^{(i-1)}\paren[\big]{\blocks(\Omega)}}>
    \card[\big]{R^{(i-1)}\paren[\big]{\blocks(\Omega_i)}}$ for the
    orbits~$\Omega\in\mathcal{O}$ we need to consider, and we can proceed as in
    case~1.
   \item[Case 3:] Now suppose that
    $\card{\blocks(\Omega_{\max})}=\card{\blocks(\Omega_i)}$ and that $\Omega_i$~is still an
    orbit of~$R^{(i)}$. Consider any $\Omega\in\mathcal{O}$ with
    $\card{\blocks(\Omega)}=\card{\blocks(\Omega_{\max})}$ that is still an orbit
    of~$R^{(i-1)}$. If $\Omega$~is not linked to~$\Omega_i$ (this includes~$\Omega_{\max}$ and
    all orbits linked to the latter), then Lemma~\ref{lem:linked-or-alt-product} implies
    $\Alt\paren[\big]{\blocks(\Omega)}\times\Alt\paren[\big]{\blocks(\Omega_i)}\le
    G\paren[\big]{\blocks(\Omega)\times\blocks(\Omega_i)}$. Thus
    $\Alt\paren[\big]{\blocks(\Omega)}\le R^{(i)}\paren[\big]{\blocks(\Omega)}$. On the
    other hand, if $\Omega$~and~$\Omega_i$ are linked in~$G$ (and thus also
    in~$R^{(i-1)}$), then Corollary~\ref{cor:linked} implies that setwise
    stabilizing $\Delta\subsetneq \Omega_i$ is equivalent to stabilizing a block in
    the orbit~$\Omega$, which thus is no longer an orbit of~$R^{(i)}$.
  \end{description}
  Applying Corollary~\ref{cor:linked} repeatedly to
  $R^{(k)}\paren[\big]{\bigtimes_{i=k+1}^{\ell}\blocks(\Omega_i)}$, we can obtain a
  set $\mathcal{T}'\subseteq\blocks(\Omega_{\max})$ with $R=R^{(k)}_{\mathcal{T}'}$
  and $\card{\mathcal{T}'}\le\card{\mathcal{T}}$. Moreover,
  $\Alt\paren[\big]{\blocks(\Omega_{\max})}\le
  R^{(k)}\paren[\big]{\blocks(\Omega_{\max})}$ implies $\Alt(\mathcal{S})\le
  R(\mathcal{S})$ for $\mathcal{S}=\blocks(\Omega_{\max})\setminus\mathcal{T}'$.
  Note that $\card{\mathcal{S}}\ge\card{\blocks(\Omega_{\max})}-\card{T}>k$. Thus
  $\Omega_{R}=\bigcup\mathcal{S}$ is the largest orbit of~$R$ that is contained
  in~$\Omega_{\max}$ and $\mathcal{B}_{R}=\mathcal{S}$, proving the claim.

  \begin{claim}
    $\mathcal{B}_H$ is a maximal block system for the orbit~$\Omega_{H}$ in~$H$.
    Moreover, $\card{\mathcal{B}_{H}}>k$ and $\Alt(\mathcal{B}_H)\le
    H(\mathcal{B}_H)$.
  \end{claim}
  Let $\mathcal{T}=\set{\Delta_1,\dotsc,\Delta_m}$ be an enumeration of the
  blocks with vertices from~$T$. Consider the sequence of subgroups
  $H=H^{(m)}\le\dotsm\le H^{(0)}=R$, where $H^{(i)}=H^{(i-1)}_{[T\cap\Delta_i]}$.
  As $\card{\mathcal{B}_R}>k>\card{\Delta_i}$, Lemma~\ref{lem:bigger-survives} can be
  applied to $H^{(i-1)}\paren[\big]{\mathcal{B}_R\times\Delta_i}$. It follows
  that $\Alt(\mathcal{B}_R)\le H^{(i)}(\mathcal{B}_R)$. Thus we get $\Omega_H=\Omega_R$
  and $\mathcal{B}_H=\mathcal{B}_R$, and the claim is shown.

  The following claim concludes the correctness proof.
  \begin{claim}
    Let~$G$ and~$\Delta$ be as in the algorithm on line~\ref{line:delta}. Then
    for any $\pi\in G$ of weight~$k$ that satisfies~$F$, there is a $\pi'\in
    G_{\set{\Delta}}$ of weight~$k$ that satisfies~$F$.
  \end{claim}
  Choose $\Delta'\in\mathcal{B}_H$ with $\Delta'\cap\support(\pi)=\emptyset$;
  this is possible because $\card{\mathcal{B}_H}>k$ and $\card{\support(\pi)}\le
  k$. As $\Alt(\mathcal{B}_H)\le H(\mathcal{B}_H)$, there is a $\rho\in H$ with
  $\rho(\Delta)=\Delta'$. Thus $\pi'=\rho\pi\rho^{-1}$ is in~$G_{\set{\Delta}}$.
  Clearly, conjugation preserves weight. Further, as $\rho\in H$ implies for all
  $v\in T$ that $\rho(v)=v$ and thus $\pi'(v)=\pi(v)$, we get that
  $\pi'$~satisfies~$F$. This proves the claim.

  Computing~$G$ on line~\ref{line:mincomplautos} takes
  $(dk)^{\mathcal{O}(k^2)}\poly(N)$ time by~\cite[Theorem~3.9]{AKKT17}. On
  line~\ref{line:bcc}, the call to~\algo{ColorExactCNFGA} takes
  $(kb!)^{\mathcal{O}(k^2)}k^{\mathcal{O}(\card{F})}\poly(N)$ time by
  Theorem~\ref{thm:colorcnfga}. Each iteration of the while loop increases the
  number of orbits. The same is true for all except the last iteration of the
  repeat loop, and we can attribute the time of its last iteration to the
  containing while loop. Thus the number of iterations is bounded by
  $n=\card{V}$. As all operations in the loops can be implemented in
  $\poly(n)$~time, this shows the claimed time bound of
  $\paren[\big]{d(k^k+\card{F})!}{}^{\mathcal{O}(k^2)}\poly(N)$.
\end{proof}

Now we are ready to turn to $\ExactCNFHGI$. Our algorithm uses the following
transformation on formulas.
Given a formula~$F$ over~$\Var(V)$ and $\psi\in\Sym(V)$, let~$\psi(F)$ denote
the formula obtained from~$F$ by replacing each variable~$x_{uv}$
by~$x_{u\psi(v)}$.

\begin{lemma}\label{lem:permute-formula}
  A product $\sigma=\phi\pi\in\Sym(V)$ satisfies a formula~$F$ over~$\Var(V)$ if
  and only if $\phi$ satisfies~$\pi^{-1}(F)$.
\end{lemma}
\begin{proof}
  By definition, $\sigma$~satisfies~$x_{uv}$ if and only if $\sigma(u)=v$, which
  is equivalent to $\phi(u)=\pi^{-1}(v)$, i.e., to
  $\phi$~satisfying~$\pi^{-1}(x_{uv})$.
\end{proof}

\begin{lstlisting}[caption={$\algo{Exact-CNF-HGI}_d(X_1,X_2,k,F)$},label=alg:exactcnfhgi]
Input: Two hypergraphs $X_1$ (*and*) $X_2$ on vertex set $V$ with hyperedge size bounded by $d$, a parameter $k\in\Nset$ (*and*) a CNF formula $F$ over $\Var(V)$
Output: An isomorphism $\sigma$ from $X_1$ (*to*) $X_2$ with $\card{\support(\sigma)}=k$ that satisfies $F$, (*or~$\bot$~if~none*) exists
$\pi\assign{}$some isomorphism from $X_1$ (*to*) $X_2$ with $\card{\support(\pi)}\le k$ // see \cite[Theorem~3.8]{AKKT17}\label{line:representative}
for $U\subseteq\support(\pi)$ do // we will force $u\notin\support(\phi\pi)$ for $u\in U$
    for $M\subseteq\support(\pi)\setminus U$ do // we will force $u\in\support(\phi)\cap\support(\phi\pi)$ for $u\in M$
        $I\assign\support(\pi)\setminus(U\cup M)$ // we will force $u\notin\support(\phi)$ for $u\in I$
        $F'\assign \pi^{-1}(F)\land \bigwedge_{u\in U}x_{u\pi^{-1}(u)}\land\bigwedge_{u\in M}(\lnot x_{u\pi^{-1}(u)}\land\lnot x_{u,u})\land\bigwedge_{u\in I}x_{uu}$
        $k'\assign k-\card{I}+\card{U}$
        $\phi\assign\algo{Exact-CNF-HGA}_d(X_1,k',F')$ // see Algorithm \ref{alg:exactcnfhga}
        if $\phi\neq\bot$ then return $\sigma=\phi\pi$
return $\bot$
\end{lstlisting}

\begin{theorem}\label{thm:cnfhgi}
  Algorithm~\ref{alg:exactcnfhgi} solves $\ExactCNFHGI$ in time
  $\paren[\big]{d(k^k+\card{F})!}^{\mathcal{O}(k^2)}\poly(N)$.
\end{theorem}
\begin{proof}
  Suppose Algorithm~\ref{alg:exactcnfhgi} returns a permutation $\sigma=\phi\pi$.
  Then $\pi$~is an isomorphism from~$X_1$ to~$X_2$ and $\phi$~is an automorphism
  of~$X_1$ that satisfies~$F'$ and has weight~$k'$. As
  $\phi$~satisfies~$\pi^{-1}(F)$, Lemma~\ref{lem:permute-formula} implies that
  $\sigma$~satisfies~$F$. The additional literals in~$F'$ ensure
  $\support(\sigma)=(\support(\phi)\setminus U)\cup I$ and thus
  $\card[\big]{\support(\sigma)}=k'-\card{U}+\card{I}=k$.

  Now suppose there is an isomorphism~$\sigma$ from~$X_1$ to~$X_2$ that
  satisfies~$F$ and has weight~$k$; we need to show that the algorithm does not return~$\bot$ in this case.
  Let~$\pi$ be the isomorphism computed on
  line~\ref{line:representative}. Then $\phi=\sigma\pi^{-1}$ is an automorphism
  of~$X_1$; it satisfies~$\pi^{-1}(F)$ by Lemma~\ref{lem:permute-formula}. In
  the iteration of the loops where
  $U=\set[\big]{u\in\support(\pi)\cap\support(\phi)}{u^{\phi\pi}=u}$ and
  $M=\paren[\big]{\support(\pi)\cap\support(\phi)}\setminus U$, it holds that
  $\phi$~has weight~$k'$ and satisfies~$F'$. Thus
  $\algo{Exact-CNF-HGA}_d(X_1,k',F')$ does not return~$\bot$
  in this iteration, and
  $\algo{Exact-CNF-HGI}$ does not return~$\bot$ either.

  The isomorphism~$\pi$ can be found in $(dk)^{\mathcal{O}(k^2)}\poly(N)$
  time~\cite[Theorem~3.8]{AKKT17}. The loops have at most $3^k$~iterations, and
  $\algo{Exact-CNF-HGA}_d$ takes
  $\paren[\big]{d(k^k+\card{F})!}{}^{\mathcal{O}(k^2)}\poly(N)$ time.
  The latter term thus bounds the overall runtime.
\end{proof}

\section{Constrained isomorphisms with arbitrary weight}
\label{sec:fptgi}

In this section, we show that finding graph isomorphisms with constraints and
without weight restrictions is in $\FPT^{\prob{GI}}$.

\begin{problemdef}{\CNFHGI}
  Given two hypergraphs $X_1=(V,E_1)$ and $X_2=(V,E_2)$, and a CNF~formula~$F$
  over $\Var(V)$, decide whether there is an isomorphism from~$X_1$ to~$X_2$
  that satisfies~$F$. The parameter is $\card{F}$.
\end{problemdef}

Let~$T\subseteq\Var(V)$ be the variables that occur in the given formula~$F$.
Our approach is to enumerate satisfying assignments to~$T$. We are only
interested in assignments~$\alpha\colon T\to\set{0,1}$ that are the restriction
of the assignment given by some $\sigma\in\Sym(V)$, i.e., for all $u\in V$ it
holds that $\sum_{v\in V:x_{u,v}\in T}\alpha(x_{u,v})\le 1$ and this sum is~$1$ if
$\set{x_{u,v}}{v\in T}\subseteq T$. We call an assignment to~$T$ that satisfies
these conditions a \emph{partial permutation assignment.}

When a partial permutation assignment~$\alpha$ has $\alpha(x_{u,v})=1$, this can
be easily encoded into the graph isomorphism instance using additional colors;
we call the resulting graphs $X'$~and~$Y'$. The challenge is to enforce that a
permutation complies with $\alpha(x_{u,v})=0$. In the following algorithm, we
use the inclusion-exclusion principle to count isomorphisms that avoid the set
$P=\set{(u_1,v_1),\dotsc,(u_k,v_k)}$ of forbidden pairs given by~$\alpha$. For
$S\subseteq[k]$, we compute the size~$n_S$ of the set
\[I_S=\set[\big]{\sigma\in\Iso(X',Y')}{\forall i\in S: u_i^\sigma=v_i}.\]
To do so, we encode the additional forced mappings using additional colors and
use the \prob{GI}~oracle to decide whether the resulting graphs $X_S$~and~$Y_S$ are
isomorphic (otherwise we have $n_S=0$), and if so, to compute a generating set for the
automorphism group of~$X_S$, whose size then gives~$n_S$. Note that there is
an isomorphism compatible with~$\alpha$ if and only if
\begin{equation}
  I_S\supsetneq\bigcup_{i\in[k]} I_{\set{i}}.\tag{*}\label{eq:inclexcl}
\end{equation}

By the inclusion-exclusion principle, the size of this union can be computed as
\[\card[\Big]{\bigcup_{i\in[k]} I_{\set{i}}}
  =\sum_{\emptyset\ne S\subseteq[k]}(-1)^{\card{S}+1}\cdot
  \card[\Big]{\bigcap_{i\in S}I_{\set{i}}}
  =\sum_{\emptyset\ne S\subseteq[k]}(-1)^{\card{S}+1}\cdot\card[\big]{I_S}.\]
This solves the decision version of \CNFHGI. To solve the search version, we
check if a tentative mapping $u\mapsto v$ leads to a solution by intersecting
both sides of~\eqref{eq:inclexcl} with $\set{\sigma\in\Sym(V)}{u^\sigma=v}$;
this restriction can again be encoded using additional colors in the oracle
queries. The resulting condition can be decided using the inclusion-exclusion
principle once again.

\begin{lstlisting}[caption={$\algo{CNF-HGI}(X,Y,F)$},label=alg:cnf-hgi]
Input: Two hypergraphs $X$ (*and*) $Y$ on vertex set $V$ (*and*) a CNF formula $F$ over $\Var(V)$
Output: (*An isomorphism~$\sigma$ from~$X$ to~$Y$ that satisfies $F$, or $\bot$ if none exists*)
$T\assign{}$the variables in $F$
for each partial permutation assignment $\alpha\colon T\to\set{0,1}$ do // at most $2^{\card{T}}\le2^{\card{F}}$ iterations
    $X'\assign X; Y'\assign X; P\assign\emptyset$
    for $x_{u,v}\in T$ do
        if $\alpha(x_{u,v})=1$ then
            (*give color $u$ to $u$ in $X'$ and to $v$ in $Y'$*)
        else if $\nexists v'\in V: x_{u,v'}\in T\land \alpha(x_{u,v'})=1$ then
            $P\assign P\cup\set{(u,v)}$ // forbidden pair
    let $\set{(u_1,v_1),\dotsc,(u_k,v_k)}=P$
    for $S\subseteq[k]$ do // $2^{\card{P}}\le2^{\card{F}}$ iterations
        $n_S\assign\card{\set{\sigma\in\Iso(X',Y')}{\forall i\in S: u_i^\sigma=v_i}}$ // use the \prob{GI} oracle and additional colors
    if $n_\emptyset>\sum_{\emptyset\ne S\subseteq[k]}(-1)^{\card{S}+1}\cdot n_S$ then // there is a solution that is compatible with $\alpha$
        while there is a vertex $u$ in $X'$ whose color is not unique do
            for each vertex $v$ in $Y'$ that has the same color as $u$ in $X'$ do
                for $S\subseteq[k]$ do
                    $n_{u,v,S}\assign\card{\set{\sigma\in\Iso(X',Y')}{u^\sigma=v\land\forall i\in S: u_i^\sigma=v_i}}$
                if $n_{u,v,\emptyset}>\sum_{\emptyset\ne S\subseteq[k]}(-1)^{\card{S}+1}\cdot n_{u,v,S}$ then
                    (*give color $u$ to $u$ in $X'$ and to $v$ in $Y'$*)
                    (*continue with the next iteration of the while-loop*)
        return (*the unique color-preserving isomorphism from $X'$ to $Y'$*)
return $\bot$
\end{lstlisting}

\begin{theorem}
  Algorithm~\ref{alg:cnf-hgi} solves \CNFHGI in $2^{2\card{F}}\poly(N)$ time when
  given access to a $\prob{GI}$~oracle; the oracle queries have size $\poly(N)$.
\end{theorem}

\begin{proof}
  The correctness follows from the observations above. Regarding the time bound,
  the outer for-loop incurs a cost of~$2^{\card{T}}\le2^{\card{F}}$. The
  for-loops over $S\subseteq[k]$ and the sums are not nested and contribute
  factor of $2^{\card{P}}\le2^{\card{F}}$. The remaining loops and operations
  are polynomial in the input size when given a \prob{GI}~oracle.
\end{proof}

\section{Exact complexity}\label{sec:exact-complexity}

The complexity of a permutation $\pi\in\Sym(V)$ can be bounded by functions of
its weight: $\card[\big]{\support(\pi)}-1\le\compl(\pi)\le
2\cdot\card[\big]{\support(\pi)}$. However, there is no direct functional
dependence between these two parameters. And while the algorithms of
Sections~\ref{Sec:bcc}~and~\ref{sec:exactgi} can be modified to find
isomorphisms of exactly prescribed complexity, we give an independent and more
efficient algorithm in this section.

The main ingredient is an analysis of decompositions
$\sigma=\sigma_1\dotsm\sigma_\ell$ of $\sigma\in\Sym(V)$ into
$\sigma_i\in\Sym(V)\setminus\{\id\}$ (for $1\le i\le\ell$) with
$\compl(\sigma)=\sum_{i=1}^\ell\compl(\sigma_i)$; we call such decompositions
\emph{complexity-additive}. For example, the decomposition into
complexity-minimal permutations provided by Lemma~\ref{lem:decompose-compl} is
complexity-additive.

For a sequence of permutations $\sigma_1,\dotsc,\sigma_\ell\in\Sym(V)$, its
\emph{cycle graph} $\CG(\sigma_1,\dotsc,\sigma_\ell)$ is the incidence graph
between $\bigcup_{i=1}^\ell\support(\sigma_i)$ and the $\sigma_i$-orbits of size
at least~$2$, i.e., the cycles of~$\sigma_i$, for $1\le i\le\ell$. We call the
former \emph{primal vertices} and the latter \emph{cycle-vertices}.

\begin{figure}
  \centering
  \begin{tikzpicture}[scale=.9,rotate=30,label position=15,
    every label/.style={rectangle,draw=none,fill=none,font=\scriptsize,inner sep=1pt}]
    \path[every node/.style={draw, circle, inner sep=2pt}]
    (0,1.5)     node[label=0] (0) {0}
    (.866,0)    node[label=1] (1) {1}
    (1.732,1.5) node[label=2] (2) {2}
    (1.732,-.5) node[label=3] (3) {3}
    (3.464,.5)  node[label=4] (4) {4}
    (2.598,-1)  node[label=5] (5) {5}
    (4.330,-1)  node[label=6] (6) {6}
    (5.196,-.5) node[label=7] (7) {7}
    (4.330,-2)  node[label=8] (8) {8}
    (6.062,-2)  node[label=9] (9) {9};
    \path[every node/.style={fill, circle, inner sep=2pt},
    every edge/.style={draw,shorten <=1pt,shorten >=1pt}]
    (.866,1)    node[label=1] (012) {} edge (0) edge (1) edge (2)
    (1.732,.5)  node[label=2] (23)  {} edge (2) edge (3)
    (2.598,1)   node[label=3] (24)  {} edge (2) edge (4)
    (3.464,-.5) node[label=1] (456) {} edge (4) edge (5) edge (6)
    (5.169,-1.5)node[label=3] (789) {} edge (7) edge (8) edge (9);
  \end{tikzpicture}
  \caption{The colored cycle graph
    $\CG_{\id}\paren[\big]{(0,1,2)(4,5,6),(2,3),(2,4)(7,8,9)}$; the colors are
    depicted next to the vertices.}\label{fig:cycle-graph}
\end{figure}
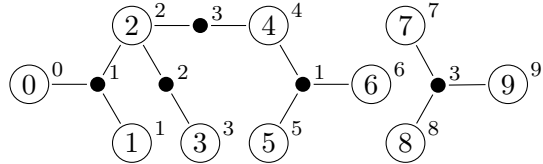

\begin{lemma}\label{lem:cycle-forest}
  Let $\sigma\in\Sym(V)$ and let $\sigma=\sigma_1\dotsm\sigma_\ell$ be a
  complexity-additive decomposition. Then $\CG(\sigma_1,\dotsc,\sigma_\ell)$ is
  a forest.
\end{lemma}
\begin{proof}
  If any of the permutations~$\sigma_i$ has more than one cycle, we further
  decompose it into its cycles. Note that this does not change the cycle graph.
  For the rest of this proof, we assume that each permutation~$\sigma_i$ has a
  single nontrivial orbit~$\Omega_i$.

  For the sake of contradiction, we assume that the graph
  $\CG(\sigma_1,\dotsc,\sigma_\ell)$ is no forest. Let~$\ell'$ be the smallest
  index such that $\CG(\sigma_1,\dotsc,\sigma_{\ell'})$ is no forest. We first
  consider the case $\ell'=2$, i.e., that $\Omega_1$~and~$\Omega_2$ have $j\ge2$
  points in common. Note that $\compl(\sigma)=\sum_{i=1}^\ell\compl(\sigma_i)$
  implies $\compl(\sigma_1\sigma_2)=\compl(\sigma_1)+\compl(\sigma_2)$. By the
  definition of complexity, we get
  \begin{align*}
    \compl(\sigma_1)+\compl(\sigma_2)&=\card{\Omega_1}-1+\card{\Omega_2}-1\text{ and}\\
    \compl(\sigma_1\sigma_2)&=\card{\Omega_1}+\card{\Omega_2}-j-c,
  \end{align*}
  where $c$~is the number of cycles of~$\sigma_1\sigma_2$. This yields a
  contradiction, as $j\ge2$ and $c\ge 1$.

  If $\ell'>2$, we know that $X=\CG(\sigma_1,\dotsc,\sigma_{\ell'-1})$ contains
  a path between two points in~$\Omega_{\ell'}$; let us call them $u$~and~$v$.
  As $X$~is a forest, $u$~and~$v$ are in the same cycle
  of~$\phi=\sigma_1\dotsm\sigma_{\ell'-1}$. (This follows by induction on the
  size of the connected component, as the product of two cycles that share a
  single point is again a cycle.) Let $\phi=\phi_1\dotsm\phi_j$ be a
  decomposition of~$\phi$ into its cycles such that $u,v\in\support(\phi_1)$,
  and let $\phi'_i=\sigma_{\ell'}^{-1}\phi_i\sigma_{\ell'}$. Then we get the
  decomposition
  \begin{equation}
    \tag{*}\label{eq:cycle-forest-decomp}
    \sigma=\phi_1\sigma_{\ell'}\phi'_2\dotsm\phi'_j\sigma_{\ell'+1}\dotsm\sigma_\ell.
  \end{equation}

  Note that each factor of this decomposition is a single cycle, because
  conjugation preserves cycle structure. As
  $\sum_{i=1}^{j}\compl(\phi_i)=\compl(\sigma_1\dotsm\sigma_{\ell'-1})$ and
  $\compl(\phi'_i)=\compl(\phi_i)$, the
  decomposition~\eqref{eq:cycle-forest-decomp} is complexity-additive. Moreover,
  the graph $\CG(\phi_1,\sigma_{\ell'})$ is no forest, as
  $u,v\in\support(\phi_1)\cap\support(\sigma_{\ell'})$. Thus we get the same
  contradiction as in the case $\ell'=2$.
\end{proof}

Given a complexity-additive decomposition $\sigma=\sigma_1\dotsm\sigma_\ell$ of
a permutation $\sigma\in\Sym(V)$ and a coloring $c\colon V\to[k]$, the
\emph{colored cycle graph} $\CG_c(\sigma_1,\dotsc,\sigma_\ell)$ is obtained from
the cycle graph $\CG(\sigma_1,\dotsc,\sigma_\ell)$ by coloring each primal
vertex $v\in V$ by~$c(v)$, and coloring each cycle-vertex that corresponds to a
cycle of~$\sigma_i$ by~$i$. (Note that a vertex of this graph is a cycle-vertex
if and only if it has odd distance to some leaf.) See
Figure~\ref{fig:cycle-graph} for an example.

A \emph{cycle pattern}~$P$ is a colored cycle graph
$\CG_c(\sigma_1,\dotsc,\sigma_\ell)$ where all primal vertices have different
colors. A complexity-additive decomposition
$\sigma'=\sigma'_1\dotsm\sigma'_\ell$ of a permutation~$\sigma'\in\Sym(V)$
\emph{matches} the cycle pattern~$P$ if there is a color-preserving
isomorphism~$\phi$ from $\CG_{c'}(\sigma'_1,\dotsc,\sigma'_\ell)$ to~$P$ for
some coloring $c'\colon V\to[k]$. Similarly, this decomposition
\emph{weakly matches}~$P$ if there is a coloring $c'\colon V\to[k]$ and a
surjective color-preserving homomorphism~$\phi$ from
$\CG_{c'}(\sigma'_1,\dotsc,\sigma'_\ell)$ to~$P$ where $\phi(u)=\phi(v)$ for
$u\neq v$ implies that $u$~and~$v$ both belong to~$V$ and are in different
$\sigma'$-orbits.

\begin{lemma}\label{lem:cycle-patterns}
  Let $P=(U,E)$ be a forest with vertex coloring~$c\colon U\to[k]$ such that
  \begin{enumerate}[nosep,label=(\arabic*)]
   \item $P$ contains no isolated vertices,
   \item the set~$V$ of vertices that have even distance to some leaf and the
    set~$C$ of vertices that have odd distance to some leaf are disjoint,
   \item the restriction~$c'$ of the coloring~$c$ to~$V$ is injective,
   \item $\set{c(u)}{u\in C}=[\ell]$ for some $\ell\in\Nset$, and
   \item any two vertices $u,v\in C$ with $c(u)=c(v)$ have distance more
    than~$2$.
  \end{enumerate}
  Then $P$~is a cycle pattern. Moreover, there is $\sigma_P\in\Sym(V)$ and a
  decomposition $\sigma_P=\sigma_1\dotsm\sigma_\ell$ that matches~$P$.
\end{lemma}
Note that any cycle pattern satisfies the properties (1)~to~(5).
\begin{proof}
  For $u\in C$, let~$\pi_u$ be a cycle on the neighbors of~$u$ in~$P$. For
  $i\in[\ell]$, let~$\sigma_i$ be a product of the cycles $\set{\pi_u}{u\in C,
  c(u)=i}$. Note that the order of the multiplication does not matter, as these
  cycles are disjoint because of property~(5). Then the colored cycle graph
  $\CG_{c'}(\sigma_1,\dotsc,\sigma_\ell)$ is isomorphic to~$P$ via the
  isomorphism~$\phi$ that maps each vertex from~$V$ to itself and the vertex of each
  cycle~$\pi_u$~of some~$\sigma_i$ to~$u$.
\end{proof}

\begin{lemma}\label{lem:pattern-complexity}
  Let $P$ be a cycle pattern. Then for any $\sigma\in\Sym(V)$ that has a
  complexity-additive decomposition $\sigma=\sigma_1\dotsm\sigma_\ell$ that
  weakly matches~$P$, it holds that $\compl(\sigma)=\compl(\sigma_P)$, where
  $\sigma_P$ is the permutation given by Lemma~\ref{lem:cycle-patterns}.
\end{lemma}
\begin{proof}
  Let $\sigma_P=\sigma'_1\dotsm\sigma'_{\ell'}$ be the decomposition
  of~$\sigma_P$ given above; as it matches~$P$, there is a color-preserving
  isomorphism~$\psi$ from~$P$ to $\CG_{c'}(\sigma'_1,\dotsc,\sigma'_\ell)$ for
  some coloring~$c'\colon V\to[k']$. As $\sigma=\sigma_1\dotsm\sigma_\ell$
  weakly matches~$P$, there is a coloring $c\colon V\to[k]$ and a surjective
  color-preserving homomorphism~$\phi$ from
  $\CG_{c}(\sigma_1,\dotsc,\sigma_{\ell'})$ to~$P$ such that $\phi(u)=\phi(v)$
  for $u\neq v$ implies that $u$~and~$v$ both belong to~$V$ and are in different
  $\sigma$-orbits. Restricting~$\phi\psi$ to the cycle-vertices yields a
  bijection from the cycles of $\sigma_1,\dotsc,\sigma_\ell$ to the cycles of
  $\sigma'_1,\dotsc,\sigma'_{\ell'}$. Because of the coloring it follows that
  $\ell'=\ell$ and that, for $i\in[\ell]$, the restriction of~$\phi\psi$ to the
  cycles of~$\sigma_i$ is a bijection to the cycles of~$\sigma'_i$. Let~$\Omega$
  be a cycle of~$\sigma_i$. As Lemma~\ref{lem:cycle-forest} implies that all
  elements of~$\Omega$ are in the same $\sigma$-orbit, and as $\phi$ is
  surjective and only allowed to identify vertices from different
  $\sigma$-orbits, the degree of~$\Omega$ in
  $\CG_c(\sigma_1,\dotsc,\sigma_\ell)$ equals the degree of~$\phi(\Omega)$
  in~$P$ and thus also the degree of~$\phi\psi(\Omega)$
  in~$\CG_{c'}(\sigma'_1,\dotsc,\sigma'_\ell)$. This implies
  $\card{\phi\psi(\Omega)}=\card{\Omega}$ and thus
  $\compl(\sigma_i)=\compl(\sigma'_i)$, which in turn implies
  $\compl(\sigma)=\compl(\sigma_P)$.
\end{proof}

\begin{lemma}\label{lem:small-pattern}
  Let $\sigma\in\Sym(V)$. Then any complexity-additive decomposition
  $\sigma=\sigma_1\dotsm\sigma_\ell$ (weakly) matches some cycle pattern that
  has at most $3\cdot\compl(\sigma)$~vertices.
\end{lemma}
\begin{proof}
  Let $t=\compl(\sigma)$. We have $k=\card{\support(\sigma)}\le2t$. Let $c\colon
  V\to[k]$ be a coloring whose restriction to~$\support(\sigma)$ is injective.
  Then $P=\CG_c(\sigma_1,\dotsc,\sigma_\ell)$ is a pattern, as the
  vertices in $V\setminus\support(\sigma)$ are isolated in
  $\CG(\sigma_1,\dotsc,\sigma_\ell)$ by Lemma~\ref{lem:cycle-forest} and thus
  are not contained in~$P$. The same lemma also implies that
  $\sigma_1,\dotsc,\sigma_\ell$ together have at most~$t$ cycles. Thus $P$~has
  at most $3t$~vertices. It remains to observe that $\phi=\id$ is an isomorphism
  from $\CG_c(\sigma_1,\dotsc,\sigma_\ell)$ to~$P$.
\end{proof}

As there are less than~$\mathcal{O}\paren[\big]{(3t)^{3t}}$ forests on $3t$~vertices and
$3t^{2t}$~ways to color them using $2t$~colors, Lemmas~\ref{lem:cycle-patterns},
\ref{lem:pattern-complexity}~and~\ref{lem:small-pattern} imply the following.
\begin{corollary}\label{cor:pattern-sets}
  For any $t\in\Nset$, there is a set~$\mathcal{P}_t$ of
  $t^{\mathcal{O}(t)}$~cycle patterns such that a permutation~$\sigma\in\Sym(V)$
  has complexity~$t$ if and only if it has a complexity-additive decomposition
  $\sigma=\sigma_1\dotsm\sigma_\ell$ that weakly matches a pattern
  in~$\mathcal{P}_t$. Moreover, $\mathcal{P}_t$ can be computed in
  $t^{\mathcal{O}(t)}$~time.
\end{corollary}
For a pattern~$P$, let~$P_i$ denote the subgraph of~$P$ induced by the
cycle-vertices of color~$i$ and their neighbors. A
permutation~$\sigma\in\Sym(V)$ and a coloring $c\colon V\to[k]$ \emph{realize
  color~$i$ of~$P$} if there is an isomorphism~$\phi$ from~$\CG_c(\sigma)$
to~$P_i$ that preserves colors of primal vertices.

\begin{lstlisting}[caption={$\algo{ExactComplexityIso}_d(X,Y,t)$},label=alg:exact-complexity]
Input: Two hypergraphs $X$ (*and*) $Y$ on vertex set $V$ with hyperedge size bounded by $d$, (*and*) $t\in\Nset$
Output: (*An isomorphism~$\sigma$ from~$X$ to~$Y$ with $\compl(\sigma)=t$, or $\bot$ if none exists*)
$A\assign\set[\big]{\sigma\in\Aut(X)}{\sigma\text{ has minimal complexity in }\Aut(X)\text{ and }\card{\support(\sigma)}\le 2t}$ //~\mbox{see~\cite[Algorithm~3]{AKKT17}}
if $X=Y$ then $I\assign A$ else
    $I\assign\set[\big]{\sigma\in\Iso(X,Y)}{\sigma\text{ has minimal complexity in }\Iso(X,Y)\text{ and }\card{\support(\sigma)}\le 2t}$ //~\mbox{see~\cite[Algorithm~2]{AKKT17}}
for $P\in\mathcal{P}_t$ do // see Corollary~\ref{cor:pattern-sets}
    $k\assign{}$the number of primal vertices in $P$
    $\ell\assign{}$the number of colors of (*cycle-vertices*) in $P$
    for $h\in\mathcal{H}_{V,k}$ do // $\mathcal{H}_{V,k}$ is the perfect family of hash functions $h\colon V\to [k]$ from~\cite{FKS84}
        if there are $\sigma_1,\dotsc,\sigma_{\ell-1}\in A$ (*and*) $\sigma_\ell\in I$ s.t. $(\sigma_i,h)$ realize color $i$ of $P$ then
            return $\sigma=\sigma_1\dotsm\sigma_\ell$
return $\bot$
\end{lstlisting}

\begin{theorem}\label{thm:exact-complexity}
  Given two hypergraphs~$X$~and~$Y$ of hyperedge size at most~$d$ and
  $t\in\Nset$, the algorithm $\algo{ExactComplexityIso}_d(X,Y,t)$ finds
  $\sigma\in\Iso(X,Y)$ with $\compl(\sigma)=t$ (or determines that there is
  none) in $\mathcal{O}\paren[\big]{(dt)^{\mathcal{O}(t^2)}\poly(N)}$ time.
\end{theorem}
\begin{proof}
  Suppose there is some $\sigma\in\Iso(X,Y)$ with $\compl(\sigma)=t$.
  Lemma~\ref{lem:decompose-compl} gives the complexity-additive decomposition
  $\sigma=\sigma_1\dotsm\sigma_\ell$ into minimal-complexity permutations
  $\sigma_1,\dotsc,\sigma_{\ell-1}\in\Aut(X)$ and $\sigma_\ell\in\Iso(X,Y)$; all
  of them have complexity at most~$t$.
  By the correctness of the algorithms from~\cite{AKKT17}, we have
  $\sigma_1,\dotsc,\sigma_{\ell-1}\in A$ and $\sigma_\ell\in I$. As
  $\mathcal{H}_{V,k}$ is a perfect hash family, it contains some function~$h$
  whose restriction to~$\support(\sigma)$ is injective. Then
  $\CG_h(\sigma_1,\dotsc,\sigma_\ell)$ is isomorphic to some $P\in\mathcal{P}_t$
  by Corollary~\ref{cor:pattern-sets}. Thus $(\sigma_i,h)$ realize color~$i$
  of~$P$, for $1\le i\le\ell$, so the algorithm does not return~$\bot$.

  Now suppose that the algorithm returns $\sigma=\sigma_1\dotsm\sigma_\ell$ with
  $\sigma_1,\dotsc,\sigma_{\ell-1}\in A\subseteq\Aut(X)$ and $\sigma_\ell\in
  I\subseteq\Iso(X,Y)$. This clearly implies $\sigma\in\Iso(X,Y)$. To show
  $\compl(\sigma)=t$, we observe that the algorithm only returns~$\sigma$ if
  there is a pattern $P\in\mathcal{P}_t$ whose cycle-vertices have
  $\ell$~colors and that contains $k$~primal vertices, and a hash function
  $h\in\mathcal{H}_{V,k}$ such that $(\sigma_i,h)$ realize color~$i$ of~$P$, for
  $1\le i\le\ell$. In particular, there is an isomorphism~$\phi_i$
  from~$\CG_h(\sigma_i)$ to~$P_i$ that preserves colors of primal vertices. As
  the primal vertices of~$P$ all have different colors and as $P$~is a forest by
  Lemma~\ref{lem:cycle-forest}, it follows that the decomposition
  $\sigma=\sigma_1\dotsm\sigma_\ell$ is complexity-additive. Now consider the
  function $\phi=\bigcup_{i=1}^\ell\phi_i$; it is well-defined, as
  $v\in\support(\phi_i)\cap\support(\phi_j)$ implies $\phi_i(v)=\phi_j(v)$
  because $P$~contains only one primal vertex of color~$h(v)$. It is surjective,
  as every vertex of~$P$ occurs in at least one~$P_i$. It is a homomorphism from
  $P_\sigma=\CG_h(\sigma_1,\dotsc,\sigma_\ell)$ to~$P$, as every edge occurs in
  the support of one of the isomorphisms~$\phi_i$. Also, $\phi(u)=\phi(v)$ for
  $u\ne v$ implies that $u$~and~$v$ are in different connected components
  of~$P_\sigma$, as $P$~is an forest; consequently $u$~and~$v$ are in different
  orbits of~$\sigma$. Thus $\sigma=\sigma_1\dotsm\sigma_\ell$ weakly
  matches~$P$. By Lemma~\ref{lem:pattern-complexity} it follows that
  $\compl(\sigma)=t$.

  It remains to analyze the runtime. The algorithms used to compute~$A$~and~$I$ each take
  $\mathcal{O}\paren[\big]{(dt)^{\mathcal{O}(t^2)}\poly(N)}$ time~\cite{AKKT17}.
  The pattern set~$\mathcal{P}_t$ can be computed in $t^{\mathcal{O}(t)}$~time
  by Corollary~\ref{cor:pattern-sets}.
  As $k\le2t$, the perfect hash
  family~$\mathcal{H}_{V,K}$ has size $2^{\mathcal{O}(t)}\log^2n$. As $\ell\le
  t$, this gives a total runtime of
  $\mathcal{O}\paren[\big]{(dt)^{\mathcal{O}(t^2)}\poly(N)}$.
\end{proof}

\section{Colored Graph Automorphism}\label{sec:red-blue}

In~\cite{AKKT17} we showed that the following parameterized version of Graph
Automorphism is \Wone-hard. It was first defined in~\cite{DoFe} and is a
generalization of the problem studied by Schweitzer~\cite{Schw}.

\begin{problemdef}{$\colGA$}
  Given a graph~$X$ with its vertex set partitioned as $\RED\cup\BLUE$, and a
  parameter~$k$, decide if there is a partition-preserving automorphism that
  moves exactly~$k$ \BLUE vertices.
\end{problemdef}

For an automorphism $\pi\in\Aut(X)$, we will refer to the number of \BLUE
vertices moved by~$\pi$ as the \BLUE~weight of~$\pi$.
The graphs used in the \Wone-hardness reduction in~\cite{AKKT17} are designed to
simulate the Circuit Value Problem for Boolean inputs of Hamming weight~$k$.
\BLUE vertices are used at the input level and are partitioned into color
classes of size~$2$ (the pair of nodes in each color class can flip or not to
simulate a Boolean value). Vertices in the graph gadgets used for simulating the
circuit gates are the \RED vertices. It turns out that in the \RED part the
color classes are of size at most~4.
In this section, we show that $\colGA$ is in~\FPT when restricted to colored
graphs where the \RED color classes have size at most~$3$.

Given an input instance $X=(V,E)$ with vertex partition $V=\RED\cup\BLUE$ such
that \RED is refined into color classes of size at most~$3$ each, our
algorithm proceeds as follows.

\begin{description}

 \item[Step 1: color-refinement.]
  
  $X$ already comes with a color classification of vertices (\RED and \BLUE, and
  within \RED color classes of size at most 3 each; within \BLUE there may be
  color classes of arbitrary size). The color refinement procedure keeps
  refining the coloring in steps until no further refinement of the vertex color
  classes is possible. In a refinement step, if two vertices have identical
  colors but differently colored neighborhoods (with the multiplicities of
  colors counted), then these vertices get new different colors.

  At the end of this refinement, each color class induces a regular graph, and
  each pair of color classes induce a semiregular bipartite graph.

 \item[Step 2: local complementation.] 

  We complement the graph induced by a color class if this reduces the number of
  its edges; this does not change the automorphism group of~$X$. Similarly, we
  complement the induced bipartite graph between two color classes if this
  reduces the number of its edges.

  Now each \RED color class induces the empty graph. Similarly, for
  $b\in\{2,3\}$, the bipartite graph between any two color classes of size~$b$
  is empty or a perfect matching. (Note that this does not necessarily hold for
  $b\ge4$.) Color refinement for graphs of color class size at most~$3$ has been
  used in earlier work~\cite{IL,JKMT}.

  Let $C\subset\RED$ and $D\subset\BLUE$ be color classes after
  Step~1. Because of the complementations we have applied, $\card{C}=1$ implies
  that $X[C,D]$ is empty, and if $\card{C}\in\set{2,3}$ then $X[C,D]$ is either
  empty or the degree of each $D$-vertex in $X[C,D]$ is~$1$.

 \item[Step 3: fix vertices that cannot move.]

  For any red color class $C\subset\RED$ whose elements have more than~$k$ \BLUE
  neighbors, give different new colors to each vertex in~$C$ (because of Step~2,
  each non-isolated \RED vertex is in a color class with more than one
  vertex). Afterwards, rerun Steps~$1$~and~$2$ so we again have a stable
  coloring.

  Fixing the vertices in~$C$ does not lose any automorphism of~$X$ that has
  \BLUE-weight at most~$k$. Indeed, as every \BLUE vertex has at most one
  neighbor in~$C$, any automorphism that moves some~$v\in C$ has to move all
  (more than~$k$) \BLUE neighbors of~$v$.

 \item[Step 4: remove edges in the red part.]

  We already observed that each \RED color class induces the empty graph.
  Let~$\mathcal{X}$ be the graph whose vertices are the \RED color classes,
  where two of them are adjacent iff there is a perfect matching between them
  in~$X$. For each $b\in\{1,2,3\}$, the \RED color classes of size~$b$ get
  partitioned into components of~$\mathcal{X}$.
  
  We consider each connected component~$\mathcal{C}$ of~$\mathcal{X}$ that
  consists of more than one color class. Let~$X'$ be the subgraph of~$X$ induced
  by vertices in~$\bigcup\mathcal{C}$ and their neighbors in~\BLUE. Because of
  Step~3, the graph~$X'$ has color class size at most~$3k$, so we can compute
  its automorphism group $H=\Aut(X')$ in $2^{\mathcal{O}(k^2)}\poly(N)$
  time~\cite{FHL80}. We distinguish several cases based on the action of~$H$ on
  an arbitrary color class $C\in\mathcal{C}$:
  \begin{description}
   \item[Case 1:] If $H(C)$ is not transitive, we split the color class~$C$ into
    the orbits of~$H(C)$ and start over with Step~1.
   \item[Case 2:] If $H(C)=\Sym(C)$, we drop all vertices
    in~$\paren[\big]{\bigcup\mathcal{C}}\setminus C$ from~$X$. And for each
    \BLUE color class~$D$ that has neighbors in at least one $C'\in\mathcal{C}$,
    we replace the edges between a vertex~$u\in D$ and~$\bigcup\mathcal{C}$ by
    the single edge $(u,v)$, where $v$~is the vertex in~$C$ that is reachable
    via the matching edges from the neighbor of~$u$ in~$C'$.
   \item[Case 3:] If $H(C)$ is generated by a $3$-cycle $(v_1v_2v_3)$, we first
    proceed as in Case~2. Additionally, we add directed edges within each \BLUE
    color class~$D$ that now has neighbors in~$C$. Let $D_i\subset D$ be the
    neighbors of~$v_i$. We add \emph{directed edges} from all vertices in~$D_i$
    to all vertices in~$D_{(i+1)\bmod 3}$ and color these directed edges
    by~$C$.
  \end{description}

  After this step, there are no edges induced on the \RED part of~$X$. Moreover,
  we have not changed the automorphisms on the induced subgraph, so the modified
  graph~$X$ still has the same automorphism group as before.

 \item[Step 5: turn red vertices into hyperedges.]

  We encode~$X$ as a hypergraph $X'=(\BLUE\cup\NEW,E')$ in which each vertex in
  \RED is encoded as a hyperedge on the vertex set $\BLUE\cup\NEW$. Let
  $\NEW=\set{v_C}{C\subset\RED\text{ is a color class}}$. Let $v\in
  C\subset\RED$ be any red vertex. We encode~$v$ as the hyperedge
  $e_v=\set[\big]{v_C}\cup\set[\big]{u\in\BLUE}{(v,u)\in E(X)}$.

  In the hypergraph $X'$ we give distinct colors to each vertex in \NEW in order
  to ensure that each color class
  $\{v_{\mathcal{C},1},v_{\mathcal{C},2},v_{\mathcal{C},3}\}$ in \RED is
  preserved by the automorphisms of~$X'$.

  Clearly, there is a 1-1 correspondence between the color-reserving
  automorphisms of~$X$ and those of~$X'$. Note that the hyperedges of~$X'$ have
  size bounded by $k+1$, as each \RED vertex in~$X$ has at most~$k$ \BLUE
  neighbors after Step~3.
  
 \item[Step 6: bounded hyperedge size automorphism.]

  We seek a weight~$k$ automorphism of~$X'$ using the algorithm
  of~\cite[Corollary~6.4]{AKKT17};\footnote{There is a caveat that in addition
    to hyperedges in the graph $X'[\BLUE]$ we also have colored directed edges.
    However, the algorithm of~\cite[Corollary~6.4]{AKKT17} needs only minor
    changes to handle this.} this is possible in
  $d^{\mathcal{O}(k)}2^{\mathcal{O}(k^2)}\poly(N)$ time.
\end{description}

\noindent
This algorithm gives us the following.

\begin{theorem}\label{thm:red-blue-gi-color-classes}
  The above algorithm solves $\colGA$ when the \RED part of the input graph is
  refined in color classes of size at most~3. It runs in
  $d^{\mathcal{O}(k)}2^{\mathcal{O}(k^2)}\poly(N)$ time.
\end{theorem}

\begin{acknowledgments}
  We thank the anonymous
  IPEC
  referees for their valuable comments.
  
  This work was supported by the Alexander von Humboldt Foundation in its
  research group linkage program. The second and third authors are supported by
  DFG grant~KO~1053/7-2. The fourth author is supported by DFG grant~TO~200/3-2.
\end{acknowledgments}

\end{document}